\documentclass[a4paper,3p,10pt,sort&compress,times]{elsarticle}
\usepackage{amsmath,amsfonts,amsthm,verbatim, amsxtra,fixltx2e,hyperref}
\DeclareMathOperator{\Ff}{Fact}

\DeclareMathOperator{\card}{card}

\DeclareMathOperator{\ext}{ext}
\DeclareMathOperator{\Ra}{Ra}
\DeclareMathOperator{\Sb}{SB}
\DeclareMathOperator{\ind}{ind}
\newcommand{\CH}{\mathrm{CH}}
\newcommand{\Lynd}{\mathrm{Lynd}}

\newcommand{\Nn}{\mathbb N}

\newcommand{\PAL}{\mathrm{PAL}}
\newcommand{\PER}{\mathrm{PER}}
\newcommand{\Stand}{\mathrm{Stand}}

\newcommand{\Aa}{\mathcal A}

\theoremstyle{plain}
\newtheorem{thm}{Theorem}[section]
\newtheorem{prop}[thm]{Proposition}
\newtheorem{cor}[thm]{Corollary}
\newtheorem{lemma}[thm]{Lemma}

\theoremstyle{definition}
\newtheorem{remark}[thm]{Remark}

\newtheorem{example}[thm]{Example}

\begin{document}

\begin{frontmatter}

\title{On Christoffel and standard words and their derivatives}
\author[dma]{Alma D'Aniello}
\ead{alma.daniello@unina.it}
\author[dma]{Aldo de Luca\fnref{indam}}
\ead{aldo.deluca@unina.it}
\author[dieti]{Alessandro De Luca\corref{cor1}\fnref{prin}}
\ead{alessandro.deluca@unina.it}

\cortext[cor1]{Corresponding author.}
\fntext[indam]{Partially supported by GNSAGA, INdAM.}
\fntext[prin]{Partially supported by the Italian Ministry of Education (MIUR), under the PRIN project number 2010LYA9RH, titled ``Automi e Linguaggi Formali: Aspetti
Matematici e Applicativi''.}
\address[dma]{Dipartimento di Matematica e Applicazioni ``R.~Caccioppoli'',
Universit\`a degli Studi di Napoli Federico II}
\address[dieti]{DIETI, Universit\`a degli Studi di Napoli Federico II}

\begin{abstract}
We introduce and study natural derivatives for  Chri\-stoffel  and  finite standard words, as well as for characteristic Sturmian words.
These derivatives, which are realized as inverse images under suitable morphisms, preserve the aforementioned classes of words. In the case of Christoffel words, the morphisms involved
map $a$ to $a^{k+1}b$ (resp.,~$ab^{k}$) and $b$ to $a^{k}b$ (resp.,~$ab^{k+1}$) for a suitable $k>0$. 
As long as derivatives are longer than one letter, higher-order derivatives are naturally obtained.
We define the
depth of a Christoffel or standard word as the smallest order for which the derivative is a
single letter. We give several combinatorial and arithmetic descriptions of the depth, and  (tight) lower and upper bounds for it.
\end{abstract}

\begin{keyword}
Christoffel word\sep Standard word\sep Central word\sep Characteristic word\sep Derivative of a word
\MSC[2010] 68R15
\end{keyword}

\end{frontmatter}

\section{Introduction}

Since the first systematic study by M. Morse and G. A. Hedlund \cite{MH}, Sturmian words have been among the most studied infinite words in combinatorics as  they are the simplest aperiodic words in terms of \emph{factor complexity}, and enjoy many beautiful characterizations and properties (see, for instance, \cite[Chap. 2]{LO2}). 

Sturmian words are of interest in several fields of mathematics  such as combinatorics, algebra, number theory,  dynamical systems, and differential equations. They are also of great importance in theoretical physics as basic  examples of $1$-dimensional quasicrystals (cf.~\cite{dB-gen} and references therein) and in  computer science where they are used in
computer graphics as digital approximation of straight lines (cf.~\cite{DSS}).

A basic tool in the study of Sturmian words is the \emph{palindromization map} $\psi$, first introduced by the second author~\cite{deluca}. It maps any finite binary word $v$ (called \emph{directive word} in this context) to a palindrome $\psi(v)$ called \emph{central word}. The definition can be naturally extended to infinite directive words; when $v$ spans among all binary words where both letters occur infinitely often, $\psi(v)$ gives exactly all \emph{characteristic} Sturmian words (or infinite standard Sturmian words). An infinite word is Sturmian if it has the same set of factors as some characteristic Sturmian word.

Central words are thus all palindromic prefixes of characteristic Sturmian words; they can also be defined in a purely combinatorial way, as words having two coprime periods $p,q$ and length $p+q-2$. If $w$ is a central word over the alphabet $\{a,b\}$, then $awb$ is a (lower) \emph{Christoffel word} and $wab,wba$ are \emph{standard words}. These classes of words, which also include the letters $a$ and $b$, represent a finite counterpart to Sturmian words and are well studied in their own right  as they satisfy remarkable and surprising combinatorial properties (see for instance~\cite{BLRS,LO2, DM}).

In a previous paper \cite{adlADL} the second and third author  have studied an important connection between the combinatorics of these words and the famous Stern sequence. In this paper, which can be considered as a continuation of the previous one,  we consider new combinatorial properties which are mainly related to the notion of derivative of a word. Word derivation, meant as inverse image under some injective morphism (also called ``desubstitution'', or ``inflation'' in~\cite{dB-Seq}), is a known topic in combinatorics on words. A well-known  instance is the notion of \emph{derivated word} of a recurrent word, introduced by F. Durand~\cite{Du} along with the important concept of \emph{return words}.

The main objective of this paper is to study some natural derivatives for noteworthy classes of finite Sturmian words, such as Christoffel and standard words.
The paper is organized as follows. In Section~\ref{sec:duebis} we consider the palindromization map. A well-known result by  J. Justin~\cite{J}, known as \emph{Justin's formula},  links the palindromization map with \emph{pure standard Sturmian morphisms}, i.e., morphisms
of the monoid $\{\mu_a, \mu_b\}^*$ where for $x\in \{a, b\}$, and $y\neq x$, $\mu_x$ is defined as follows: $\mu_x: x\mapsto x, \ y\mapsto xy$. 
Setting for $v= v_1\cdots v_n$, $\mu_v= \mu_{v_{1}}\circ\cdots\circ\mu_{v_{n}}$, one derives from Justin's formula, that every standard word $\psi(v)xy$ with $\{x,y\}=\{a,b\}$ is obtained as the image of $xy$ under the morphism $\mu_v$, 
\begin{equation}\label{eq:Jsta}
\psi(v)xy=\mu_{v}(xy)\,.
\end{equation}

In Section~\ref{sec:duetris}, some basic relations existing between central, standard, and Chri\-stoffel words are recalled and new combinatorial properties are proved.

In Section~\ref{sec:cmor},
we discuss \emph{Christoffel morphisms}, i.e., morphisms preserving Christoffel words. We provide a simple combinatorial proof for the known fact \cite{BLRS} that the monoid of Christoffel morphisms is generated by $\lambda_{a}$ and $\lambda_{b}$, defined by
\[ \lambda_a = \mu_a \ \text{ and }\ \lambda_b :a \mapsto ab, \ b \mapsto b.\]

Setting
$\lambda_{v}=\lambda_{v_{1}}\circ\cdots\circ\lambda_{v_{n}}$ for $v=v_{1}\cdots v_{n}$, this gives  an analogue of  formula~\eqref{eq:Jsta} in the case of  Christoffel words, namely 
\[a\psi(v)b=\lambda_{v}(ab).\]
 We also prove that the \emph{inverse} image of a Christoffel word under a Christoffel morphism is a Christoffel word; again, this mirrors a well-known result for standard words and morphisms.

With such knowledge about Christoffel morphisms, in Section~\ref{sec:derc} 
we define a derivative for proper Christoffel words. In fact, for each such word $w$ there exists some nonnegative integer $k$ (the \emph{index} of $w$) such that $w$ can be uniquely factored over $X_{k}=\{a^{k}b,a^{k+1}b\}$ or $Y_{k}=\{ab^{k},ab^{k+1}\}$; hence, $w$ is the image, under the morphism $\varphi_{k}=\lambda_{a^{k}b}$ or $\hat\varphi_{k}=\lambda_{b^{k}a}$, of  a word $\partial w$ that we call
the \emph{derivative} of $w$. Since $\varphi_{k}$ and $\hat\varphi_{k}$ are Christoffel morphisms, this derivative is still a Christoffel word. 

Our choice of morphisms $\varphi_{k}$ and
$\hat\varphi_{k}$ for the definition is motivated by the following  arguments. First, the factorization over $X_{k}$ or $Y_{k}$ is quite natural and has been used in well-known algorithms for recognizing factors of Sturmian words
(or \emph{digital straight segments}, in the computer graphics terminology; cf.~\cite{DSS}).
Second, if
$w=a\psi(v)b$ and $v$ is not a power of a letter, then
\[\partial w=a\psi({}_{+}v)b,\] where $_+v$ is the longest suffix of $v$ immediately preceded by a letter different from the first letter of $v$. The operator $v\mapsto{}_{+}v$ was introduced by the last two authors in~\cite{adlADL} and appears in some interesting results on Christoffel words; for instance, if $v$ starts with the letter $x$ and $\{x,y\}=\{a,b\}$, then the length $|a\psi({}_{+}v)b|= |\partial w|$
equals the number of occurrences of $y$ in $a\psi(v)b$.
Finally, 
 a Christoffel word is determined by its derivative and the value of its index.
 
 Further  results on the derivatives of Christoffel words are proved. In particular, if a Christoffel word $w$ is factored as $w=w_1w_2$
 with $w_1$ and $w_2$ proper Christoffel words, then $\partial w= \partial w_1\partial w_2$. Moreover, the length of a Christoffel word
 $w= a\psi(v_1v_2\cdots v_n)b$ with $v_i\in {\cal A}$, $1\leq i \leq n$, is equal to $2$ plus the sum of the lengths of derivatives $\partial a\psi(v_i\cdots v_n)b$, $i=1,\ldots, n$.

In Section~\ref{sec:depth}, 
we naturally define higher order derivatives, by letting $\partial^{i+1}w=\partial(\partial^{i}w)$ whenever $\partial^{i}w$ is still a proper Christoffel word (i.e., not just a letter). The \emph{depth} of a Christoffel word $w$ is then the smallest $i\geq 0$ such that $\partial^{i}w$ is a letter. We give several descriptions of the depth of $a\psi(v)b$ as a function $\delta(v)$ of its directive word. We  prove that
$\delta(uv)$ equals either $\delta(u)+\delta(v)$ or $\delta(u)+\delta(v)-1$.
Tight lower and upper bounds of the depth are given; moreover, we characterize the directive words for which such  bounds are attained.
We give also a closed formula for the number $J_k(p)$ of the words $v$ of length $k$ such that $\delta(v)=p$.

 In Section~\ref{sec:ders}
we consider  finite and infinite standard Sturmian words; using the standard morphisms $\mu_{a^{k}b}$ and $\mu_{b^{k}a}$ we define a natural derivative in these cases. This allows us to extend the previous results to standard words; in particular, the derivative of the standard word $\psi(v)xy$ with $\{x,y\}=\{a,b\}$ is either a letter or  the proper standard word $\psi(_+v)xy$, where $_+v$ is the same directive word found in the derivative of the Christoffel word $a\psi(v)b$. 
 Hence, the depths of
$\psi(v)ab$, $\psi(v)ba$, and $a\psi(v)b$ coincide. In the infinite case, the derivative $Ds$ of a characteristic Sturmian word $s$ is word isomorphic to a derivated word in the sense of Durand.
We give a proof for the fact that a characteristic Sturmian word has only finitely many distinct higher order derivatives if and only if its directive word is ultimately periodic (see also~\cite{AB}). Finally, we prove that  there exists a simple relation between the derivative $Ds$ of a characteristic word $s$ and the   derivative $\partial s$, namely $\partial s = b Ds$. 

\section{Notation and Preliminaries}\label{sec:due}

In the following, ${ A}$  will denote a finite non-empty set, or  \emph{alphabet}   and ${ A}^*$  the \emph{free monoid} generated by ${A}$. 
The elements of ${ A}$ are usually called \emph{ letters} and those of ${ A}^*$ \emph {words}. 
 The identity element
of ${ A}^*$ is called \emph{empty word} and denoted by $\varepsilon$.  We
set ${ A}^+={A}^*\setminus\{\varepsilon\}$.

A word $w\in {A}^+$ can be written uniquely as a sequence of letters as
$w=w_1w_2\cdots w_n$, with $w_i\in {A}$, $1\leq i\leq n$, $n>0$.  The
integer $n$ is called the \emph{length} of $w$ and denoted $|w|$.  The
length of $\varepsilon$ is 0.  For any $w\in { A}^*$ and $x\in { A}$, $|w|_x$
denotes the number of occurrences of the letter $x$ in $w$.
 For any word $v\in { A}^+$, we let  $v^{(F)}$ (resp., $v^{(L)}$) denote the first (resp.,  last) letter of $v$.

Let $w\in { A}^*$.  The word $u$ is a \emph{factor} of $w$ if there exist
words $r$ and $s$ such that $w=rus$.  A factor $u$ of $w$ is called
\emph{proper} if $u\neq w$.  If $w=us$, for some word $s$ (resp.,
$w=ru$, for some word $r$), then $u$ is called a \emph{prefix} (resp.,
a \emph{suffix}) of $w$. If $u$ is a prefix of $w$, then $u^{-1}w$ denotes
the word $v$ such that $uv=w$.

Let $p$ be a positive integer.  A word $w=w_1\cdots w_n$, $w_i\in {A}$,
$1\leq i\leq n$, has \emph{period} $p$ if the following condition is
satisfied: for any integers $i$ and $j$ such that $1\leq i,j\leq n$,
\[
	\text{if }i\equiv j \pmod{p} , \text{ then } w_i = w_j.
\]
Let us observe that if a word $w$ has a period $p$, then any non-empty  factor of $w$ has also the period $p$. 

We let $\pi(w)$ denote the minimal period of $w$. Conventionally, we set $\pi(\varepsilon)=1$.
A word $w$ is said to be \emph{constant} if  $\pi(w)=1$, i.e.,  $w=z^k$ with $k\geq 0$ and $z\in {A}$.
Two words $v$ and $w$ are \emph{conjugate} if there exist words $r$ and $s$ such that $v= rs$ and $w=sr$.

Let $w=w_1\cdots w_n$, $w_i\in { A}$, $1\leq i\leq n$.  The
\emph{reversal} of $w$ is the word $w\sptilde = w_n\cdots w_1$.  One
defines also $\varepsilon\sptilde =\varepsilon$.  A word is called
\emph{palindrome} if it is equal to its reversal.  We let
$\PAL$ denote the set of all palindromes on the
alphabet ${A}$.

In the following, we let the alphabet ${ A}$ be totally ordered.  We let $<_{lex}$ denote the lexicographic
order induced on ${A}^*$. A word is called a \emph{Lyndon word}  if it is lexicographically less than any of its proper suffixes (cf. \cite[Chap. 5]{LO}). As is well-known a Lyndon word $w\not\in A$ can be factored ({\em standard factorization}) as $w=lm$ where $l$ is a Lyndon word and $m$ is the longest suffix of $w$ which is a Lyndon word.

A  right-infinite word $x$, or simply \emph{infinite word}, over the alphabet ${ A}$   is just an infinite sequence of letters:
\[x=x_1x_2\cdots x_n\cdots \text{ where }x_i\in { A},\,\text{ for all } i\geq 1\,.\]
For any integer $n\geq 0$, we let $x_{[n]}$  denote the prefix $x_1x_2\cdots x_n$ of $x$ of length $n$.
A factor of $x$ is either the empty word or any sequence  $x_i\cdots x_j$ with $i\leq j$. The set of all infinite words over ${ A}$ is denoted by ${ A}^{\omega}$. An infinite word $x$ is called \emph{ultimately periodic} if there  exist words $u\in { A}^*$ and $v\in {A}^+$
such that $x= uv^{\omega}$. The word $x$ is called (purely) \emph{periodic} if $u=\varepsilon$, i.e., $x= v\cdot v \cdot v \cdots .$ A periodic word with $v\in {A}$ will be called \emph{constant}. The word $x$ is called \emph{aperiodic} if it is not ultimately periodic.

 We say that two finite or infinite words $x=x_1x_2\cdots $ and $y=y_1y_2\cdots $ on the alphabets $A$ and $A'$ respectively  are \emph{word isomorphic}, or simply \emph{isomorphic},
if there exists a bijection $\phi: A\rightarrow A'$ such that $y=\phi(x_1)\phi(x_2)\cdots.$ 

 We set ${ A}^{\infty}= { A}^*\cup {A}^{\omega}$.
For any $w\in { A}^{\infty}$  we let  $\Ff (w)$  denote
 the set  of all distinct factors
 of the word $w$.
 
 In the following, we shall mainly concern with two-letter alphabets. We let $\Aa$ denote the alphabet whose elements
 are the letters $a$ and $b$, totally ordered by setting $a<b$.

 We let $E$ denote the automorphism of  $\Aa ^*$ defined by $E(a)=b$ and
$E(b)=a$. For each $w\in \Aa ^{\infty}$, the word $E(w)$ is called the \emph{complementary} word, or simply the \emph{complement} of $w$.
We shall often use for $E(w)$ the simpler notation ${\bar w}$.

We say that a word $v\in \Aa ^k$, $k\geq 0$,  is \emph{alternating} if
for $x,y\in \Aa $ and $x\neq y$, $v = (xy)^{\frac{k}{2}}$ if $k$ is even and
$v= (xy)^{\lfloor\frac{k}{2}\rfloor} x$ if $k$ is odd, i.e., $v$ is a single letter or if $|v|>1$ any non-terminal letter in $v$ is immediately followed by its complementary.

The \emph{slope} $\eta(w)$ of a word $w\in \Aa ^+$ is the fraction $\eta(w) = \frac{|w|_b}{|w|_a}$ if  $|w|_a>0$. We set $\eta(w)= \infty$ if
$|w|_a = 0$.

 If we identify  the letters $a$ and $b$ of $\Aa $  respectively with  the digits $0$ and $1$, for each $w\in \Aa ^*$ we let
$\langle w \rangle_2$, or simply  $\langle w \rangle$, denote the \emph{standard interpretation} of $w$ as an integer at base $2$.
For instance, $\langle a\rangle=0$, $\langle b\rangle =1$, $\langle babba \rangle = 22$.

   We represent a non-empty binary word $v\in \Aa ^+$ as
\[v= x_0^{\alpha_0}\cdots  x_n^{\alpha_n},\]  where $\alpha_i \geq 1$, $x_i\in \Aa $, $0\leq i\leq n$,
  and  $x_{i+1}= {\bar x_i}$ for $0 \leq i \leq n-1$. 
 We call the list $(\alpha_0, \alpha_1,\dots,\alpha_n)$ the
\emph{integral representation} of the word $v$. Hence, the integral representation of a word $v$  and its first letter $v^{(F)}$determine uniquely $v$.  We set $\ext (v)= |x_0\cdots x_n|= n+1$ and call it \emph{extension} of $v$. Moreover, we define $\ext(\varepsilon)= 0$.

For all definitions and notation concerning words not explicitly given in the paper, the reader is referred to the book of Lothaire \cite{LO}; for Sturmian words see \cite[Chap. 2]{LO2} and \cite[Chap.s 9-10]{AS}.

\subsection{The palindromization map}\label{sec:duebis}
 We  consider  in
${ A}^*$ the operator $^{(+)} : { A}^*\rightarrow \PAL$ which
maps any word $w\in { A}^*$ into the palindrome $w^{(+)}$
defined as the shortest palindrome having the prefix $w$ (cf. \cite{deluca}).  The word 
$w^{(+)}$ is called the \emph{right palindromic closure} of $w$.  If $Q$ is the
longest palindromic suffix of $w= vQ$, then one has
\[
	w^{(+)}=vQv\sptilde \,.
\]
Let us now define the map
\[
	\psi: {A}^*\rightarrow \PAL ,
\]
called \emph{right iterated palindromic closure}, or simply \emph{palindromization map}, over ${ A}^*$,  as follows: $\psi(\varepsilon)=\varepsilon $ and for all
$u\in { A}^*$, $x\in { A} $,
\[
	\psi(ux)=(\psi(u)x)^{(+)}\,.
\]
For instance, if $u= aaba$, one has $\psi(a)= a$, $\psi(aa)= (\psi(a)a)^{(+)}= aa$, $\psi(aab)= (aab)^{(+)}=aabaa$, and 
$\psi(u)= \psi(aaba)= aabaaabaa$.

The following proposition collects some basic  properties of  the palindromization map
 (cf., for instance, \cite{ deluca, DJP}):
\begin{prop}\label{prop:basicp} The palindromization map $\psi$  satisfies the following
properties:
\begin{enumerate}[P1.]
\item The   palindromization map  is  injective.
\item  If $u$ is  a prefix of  $v$, then $\psi(u)$ is a palindromic prefix (and suffix) of $\psi(v)$.
\item If $p$ is a prefix of $\psi(w)$, then $p^{(+)}$ is a prefix of $\psi(w)$.
\item Every palindromic prefix of $\psi(v)$ is of the form $\psi(u)$ for some prefix $u$ of $v$.
\item  $|\psi(u\sptilde )|= |\psi(u)|$, for any $u\in { A}^*$.
\item The palindromization map $\psi$ over $\{a, b\}^*$  commutes  with  the automorphism $E$, i.e.,
$ \psi\circ E = E\circ \psi.$
\end{enumerate}
\end{prop}
 For any  $w\in \psi({ A}^*)$ the unique word  $u$ such that  $\psi(u)=w$ is called the \emph{directive word} of $w$.

One can extend
  $\psi$  to ${ A}^{\infty}$ defining $\psi$ on $A^{\omega}$ as follows: let  $x\in { A}^{\omega}$ be an infinite word
\[ x = x_1x_2\cdots x_n\cdots, \quad x_i\in { A}, \ i\geq 1.\]
Since by property P2 of Proposition~\ref{prop:basicp}   for all $n$, $\psi(x_{[n]})$ is a prefix of  $\psi(x_{[n+1]})$,  we  can define  the infinite word $\psi(x)$ as:
\[ \psi(x) = \lim_{n\rightarrow \infty} \psi(x_{[n]}).\]
The  map $\psi: { A}^{\omega}\rightarrow { A}^{\omega}$ is injective. The word $x$ is called the \emph{directive word} of $\psi(x)$. It has been proved in \cite{deluca} that  if  $x\in \{a,b\}^{\omega}$ the word $\psi(x)$ is a \emph{characteristic Sturmian word} (or infinite standard Sturmian word)  if and only if both the letters $a$ and
$b$ occur infinitely often in the directive word  $x$.

 \begin{example}
 Let $\Aa =\{a,b\}$. If $x = (ab)^{\omega}$, then the characteristic  Sturmian word $f=\psi( (ab)^{\omega})$ having the directive word $x$ is the famous \emph{Fibonacci word}
\[ f = abaababaabaab\cdots.\]
 If $A=\{a, b, c\}$ the word $t= \psi((abc)^{\omega})$ is the so-called \emph{Tribonacci  word}:
\[t= abacabaabacaba\cdots.\]
\end{example}

For any $x\in { A}$, we let $\mu_x$ denote  the injective endomorphism of $A^*$
defined by
\begin{equation}\label{eq:endo}
\mu_x(x)=x, \; \mu_x(y)= xy, \text{ for }  y\in { A}\setminus \{x\} .
\end{equation}
If $v=x_1x_2\cdots x_n$, with $x_i\in { A}$, $i=1,\ldots, n$, then we set:
\[ \mu_v=\mu_{x_1}\circ \cdots \circ \mu_{x_n}; \]
moreover, if $v=\varepsilon$,  $\mu_{\varepsilon}=\mathrm{id}$.
The following interesting theorem, due to  Justin \cite{J} and usually referred to as \emph{Justin's formula}, relates the palindromization map to
the morphisms $\mu_v$.

\begin{thm}\label{thm:J} For all  $ v,u\in {A}^*$
\[ \psi(vu) = \mu_v(\psi(u))\psi(v).\]
\end{thm}

An important consequence of Justin's formula is the following lemma \cite{BdD}, which will be useful in the following.

\begin{lemma}\label{Ju00} For each $w\in { A}^*$ and $v\in { A}^{\omega}$,
$\psi(wv) = \mu_w(\psi(v)).$
\end{lemma}
\noindent For instance,  if we take $w=a$, $x = (ab)^{\omega}$, then as one easily verifies
\[\psi(a (ab)^{\omega})= \mu_a(f)= aabaaabaabaaab \cdots. \]

 The case of a binary alphabet $\Aa = \{a, b\}$ deserves a special consideration.  The following remarkable proposition holds (see, for instance~\cite[Prop. 4.10]{SC}).  
 \begin{prop}\label{prop:J0000} For any $v\in \Aa ^*$ and $x,y\in \Aa $, $x\neq y$,
\[\mu_v(xy) = \psi(v)xy.\]
 \end{prop}

  \begin{cor}\label{cor:J0000C}  For any $w, v\in \Aa ^*$ and $x,y\in \Aa $, $x\neq y$,
\[\psi(wv)xy = \mu_w(\psi(v)xy).\]
  \end{cor}
  \begin{proof} By the preceding proposition one has:
\[\psi(wv)xy = \mu_{wv}(xy)= \mu_w(\mu_v(xy)) = \mu_w(\psi(v)xy). \qedhere\]
  \end{proof}
  
  \medskip

 Let $v$ be a non-empty  word. We let  $v^-$ (resp., $^-v$) denote the word obtained from $v$ by deleting the last
(resp., first) letter. If $v$ is not constant, we  let $v_+$ (resp., $_+v$) denote the longest  prefix (resp., suffix) of $v$ which is immediately
followed (resp., preceded) by the complementary of the last (resp., first) letter of $v$. For instance, if $v= abbabab$, one has
$v^- = abbaba$, $v_+= abbab$, $^-v= bbabab$, and  $_+v = babab$. From the definition one has 
\begin{equation}\label{eq:E+}
_+(E(v)) = E(_+v), \  (E(v))_+ = E(v_+), \  (_+v)\sptilde  = (v\sptilde ) _+  \ . 
 \end{equation}
 
 As shown in \cite{adlADL}, and as we shall see in some details in the next sections,  the words $v^-$, $v_+$ and $_+v$,  $^-v$  play an essential role in the combinatorics of Christoffel words.

 \begin{prop}\label{can:duebis} Let $v\in\Aa ^*$ be non-constant. Then
\[ \mu_v(a)= \mu_{v_+}(ba)= \psi(v_+)ba,  \  \mu_v(b)= \mu_{v^-}(ab)= \psi(v^-)ab, \text{ if } v^{(L)}=a \]
and
\[ \mu_v(a)= \mu_{v^-}(ba)= \psi(v^-)ba, \  \mu_v(b)= \mu_{v_+}(ab)= \psi(v_+)ab, \text{ if } v^{(L)}=b. \]
\end{prop}
\begin{proof} We shall prove the result only when $v^{(L)}=a$. The case $v^{(L)}=b$ is similarly dealt with. We can write
$v= v_+ba^r$ for a suitable $r>0$. Therefore, by Proposition~\ref{prop:J0000} one has:
\begin{equation}\label{eq:conc1}
  \mu_v(a)=   \mu_{v_+ba^r}(a) =  \mu_{v_+b}(a)=  \mu_{v_+}(ba) = \psi(v_+)ba,
  \end{equation}
and
\begin{equation}\label{eq:conc2}
 \mu_v(b)= \mu_{v^-a}(b)= \mu_{v^-}(ab)=  \psi(v^-)ab,  
\end{equation}
which proves the assertion.\end{proof}

\begin{example}
 Let  $v= abbab$. One has $v_+= abb, v^- = abba$, and $ v^{(L)}= b$. Hence,
$\psi(v_+)= ababa$, $\psi(v^-)= ababaababa$,  $\mu_{abbab}(a)= ababaabababa = \psi(v^-)ba$,
and $\mu_{abbab}(b)=ababaab= \psi(v_+)ab$.
\end{example}

An immediate consequence of Proposition~\ref{can:duebis} is the following (see also~\cite{adlADL}):
\begin{cor}\label{cor: Lynd0} Let $v\in\Aa ^*$ be non-constant,  $x$  the last letter of $v$, and $y= {\bar x}$. Then
\[\psi(v)= \psi(v_+)yx\psi(v^-) = \psi(v^-)xy \psi(v_+).\]
\end{cor}
\begin{proof}
Since $v$ is not constant, $\psi(v)xy = \mu_v(xy)= \mu_v(x)\mu_v(y)$, and by Proposition~\ref{can:duebis}, $\mu_v(x)= \psi(v_+)yx$
and $\mu_v(y)= \psi(v^-)xy$. The result follows.
\end{proof}

 \subsection{Central, standard, and Christoffel words}\label{sec:duetris}
 
 In the study of combinatorial properties of Sturmian words a crucial
role is played by the set $\PER$ of all finite words $w$ having two
periods $p$ and $q$ such that $\gcd(p,q)=1$ and $|w|= p+q-2$.

The set $\PER$ was introduced in~\cite{DM} where its main properties
were studied.  It has been proved that $\PER$ is equal
to the set of the palindromic prefixes of all standard Sturmian words, i.e.,
\[ \PER = \psi(\Aa ^*).\] 
The words of $\PER$ have been called \emph{central} in~\cite[Chap.2]{LO2}.  

The following  structural characterization of central words was
proved in~\cite{deluca} (see, also \cite{CdL}).

\begin{prop}
	\label{Prop:uno}
	A word $w$ is central if and only if $w$ is a constant
	 or it satisfies the equation:
	\[
		w=w_1abw_2=w_2baw_1
	\]
	with $w_{1},w_{2}\in \Aa ^*$.  Moreover, in this latter case, $w_1$
	and $w_2$ are central words, $p=|w_1|+2$ and $q=|w_2|+2$ are
	coprime periods of $w$, and $\min\{p,q\}$ is the minimal period of
	$w$.
\end{prop}

 The following lemma, which will be  useful in the following, is 
in~\cite{deluca}.

\begin{lemma}\label{lem:uno}
	\label{lem:lc}
	For any $w\in \PER $, one has $(wa)^{(+)},(wb)^{(+)}\in \PER $.  More
	precisely, if $w=w_1abw_2=w_2baw_1$, then
	\[
		(wa)^{(+)}=w_{2}baw_{1}abw_{2}\,,\quad
		(wb)^{(+)}=w_{1}abw_{2}baw_{1}\,.
	\]
	If $w=x^n$ with $\{x,y\}=\Aa $, then $(wx)^{(+)}=x^{n+1}$ and
	$(wy)^{(+)}=x^nyx^n$.
\end{lemma}

 Characteristic  Sturmian words can be equivalently defined in the following
way.  Let $c_0,c_1,\ldots,c_n,\ldots$ be any sequence of
integers such that $c_0\geq 0$ and $c_i>0$ for $i>0$.  We define,
inductively, the sequence of words $(s_n)_{n\geq 0}$, where
\[
	s_0=b,\ s_1=a,\text{ and }
	s_{n+1}=s_n^{c_{n-1}}s_{n-1}\text{ for }n\geq 1\,.
\]
The sequence $(s_n)_{n\geq 0}$ converges to a limit $s$ which is a
characteristic Sturmian word (cf.~\cite{LO2}).  Every characteristic
Sturmian word is obtained in this way. The Fibonacci word is obtained
when  $c_i=1$ for  $i\geq 0$.  

We let  $\Stand$ denote
the set of all the words $s_n$, $n\geq 0$ of any sequence
$(s_n)_{n\geq 0}$.  Any element  of $\Stand$ is called \emph{standard Sturmian word},
or simply \emph{standard word}.
A standard word different from a single letter is called \emph{proper}.

The following remarkable relation existing between standard and central
words has been proved in \cite{DM}:
\[ \Stand = \Aa  \cup \PER \{ab, ba\}.\]
More precisely, the following holds (see, for instance~\cite[Prop.~4.9]{SC}):
\begin{prop} \label{prop:standStu}Any proper standard word can be uniquely expressed
as $\mu_v(xy)$ with $\{x,y\}=\{a,b\}$ and $v\in \Aa ^*$. 
\end{prop}
\noindent Hence, by Proposition~\ref{prop:J0000}  one has
\[\mu_v(xy)= \psi(v)xy.\]
Let us set for any  $v\in \Aa ^*$ and $x\in\Aa $, 
\begin{equation}\label{eq:muvi}
 p_x(v)= |\mu_v(x)|.
 \end{equation} 
 From Justin's formula one derives (cf.~\cite[Prop.~3.6]{adlZ1})  that $p_x(v)$ is the minimal period of $\psi(vx)$ and then a period of $\psi(v)$. Moreover, one has (cf.~\cite[Lemma 5.1]{adlZ2})
\begin{equation}\label{eq:minimalperiod}
 p_x(v)=  \pi(\psi(vx))= \pi(\psi(v)x) 
\end{equation}
 and $\gcd(p_x(v), p_y(v))=1$, so that
\[\pi(\psi(v))= \min \{p_x(v), p_y(v)\}.\]
Moreover, if $v$ is not constant, as $v_+$ is a proper prefix of $v^-$, by Proposition~\ref{can:duebis} one derives:
\begin{equation}\label{eq:perpsi}
\pi(\psi(v))= p_{v^{(L)}}(v)= |a\psi(v_+)b|.
\end{equation}
 Since $|\mu_v(xy)| = |\mu_v(x)|+  |\mu_v(y)|$, from Proposition~\ref{prop:standStu}  and~\eqref{eq:muvi} one has
 \begin{equation}\label{eq:cent1}
 |\psi(v)|= p_x(v)+p_y(v)-2.
 \end{equation}

Let us now introduce the important notion of \emph{Christoffel word} \cite{CFF} (see also \cite{BDR}). Let $p$ and $q$ be non-negative coprime integers, and $n= p+q>0$. The (lower) Christoffel word $w$ of slope $\frac{p}{q}$ is defined as $w= x_1\cdots x_n$ with 
\[ x_i = \begin{cases}
a & \text{if $ip \bmod n > (i-1)p \bmod n$} \\
b & \text{otherwise}
\end{cases}
\]
for $i=1,\ldots, n$,  where $k \bmod n$ denotes the remainder of the Euclidean division of $k$ by $n$.
Observe that the words $a$ and $b$ are the Christoffel words with slope
$\frac{0}{1}$ and
$\infty=\frac{1}{0}$, respectively.

The Christoffel words of slope $\frac{p}{q}$ with $p$ and $q$ coprime positive integers are called \emph{proper Christoffel words}. The term slope given to the fraction $\frac{p}{q}$  is due to the circumstance that one easily derives from the definition  that $p = |w|_b$ and $q= |w|_a$.

We observe that lower Christoffel words have also an interesting geometric interpretation  in terms of suitable paths in the integer lattice $\Nn\times\Nn$ (cf.~\cite{BLRS}). It is then natural to introduce the so-called upper Christoffel words, which can also be defined similarly to lower Christoffel words, by interchanging $a$ and $b$, as well as $p$ and $q$, in the previous definition. We shall not consider these latter words in the paper, since they are simply the reversal of lower Christoffel words.

\begin{example}\label{ex:unoA} Let $p=3$ and $q=8$. The Christoffel construction is represented by  the following diagram 
\[ 0\stackrel{a}{\longrightarrow} 3\stackrel{a}{\longrightarrow} 6\stackrel{a}{\longrightarrow} 9\stackrel{b}{\longrightarrow} 1\stackrel{a}{\longrightarrow} 4\stackrel{a}{\longrightarrow} 7\stackrel{a}{\longrightarrow} 10\stackrel{b}{\longrightarrow} 2\stackrel{a}{\longrightarrow} 5\stackrel{a}{\longrightarrow} 8\stackrel{b}{\longrightarrow} 0\]
\end{example}

Let $\CH$ denote the class of Christoffel words. The following important result, proved in \cite{BDL},
shows a basic relation existing between central and Christoffel words:
\[\CH  = a\PER b \cup \Aa . \]
Moreover, one has  \cite{BL, BDL} 
\[ \CH  = \mathrm{St} \cap \Lynd , \]
where $\Lynd $ denotes the set of Lyndon words and St  the set of  (finite) factors of all Sturmian words.
Thus  CH  equals the set of all  factors of Sturmian words which are Lyndon words. The following theorem summarizes some results on Christoffel words proved in \cite{BDL,BL,BDR}.
\begin{thm}\label{thm:Lynd} Let $w$ be a proper Christoffel word. Then the following hold:
\begin{enumerate}[1.]
\item There exist and are unique two Christoffel words
$w_1$ and $w_2$ such that $w=w_1w_2$. Moreover,  $w_1<_{lex}w_2$, and $(w_1,w_2)$ is the standard factorization of $w$ in Lyndon words. 
\item If $w$ has the slope $\frac{p}{q}$, then $|w_1|=p'$,
$|w_2|=q'$, where $p'$ and $q'$ are the respective multiplicative inverse of $p$ and $q$, modulo $|w|$. 
\item Let $w= a\psi(v)b$ have the slope $ \frac{p}{q}$. Then   $p= p_a(v\sptilde )$, $q=p_b(v\sptilde )$ and 
$p'= p_a(v)$, $q'=p_b(v)$.
\end{enumerate} 
\end{thm}

\begin{example}
The Christoffel word $w$ of the Example~\ref{ex:unoA} having slope $\frac{3}{8}$ is 
\[w=aaabaaabaab =aub,\]
where $u=aabaaabaa = \psi(a^2ba)$ is the central word of length $9$ having  the two
coprime periods $p_a(v)=4$ and $p_b(v)=7$ with $v=a^2ba$. The word $w$ can be uniquely factored as $w=w_1w_2$, where $w_1$ and $w_2$ are the Lyndon words $w_1= aaab$ and $w_2= aaabaab$. One has $w_1<_{lex}  w_2$ with $|w_1|=4=p_a(v)$ and $|w_2|=7=p_b(v)$. Moreover, $w_2$ is the proper suffix of $w$ of maximal length which is a Lyndon word. Finally, 
$\psi(v\sptilde )= \psi(aba^2)= abaabaaba$, $p_a(v\sptilde )= 3= |w|_b$, $p_b(v\sptilde )=8=|w|_a$, and $|w|_bp_a(v)= 3\cdot 4=12 \equiv |w|_ap_b(v)= 8\cdot 7=56 \equiv 1 \pmod {11}$.
\end{example}

The following proposition is an immediate consequence of  item 1 of Theorem~\ref{thm:Lynd} and of Corollary~\ref{cor: Lynd0} (see also
\cite{adlADL}).

\begin{prop}\label{cor:Lyn+-} For any non-constant word $v \in\Aa^*$, the standard factorization of $a\psi(v)b$ in Lyndon words is
\[(a\psi(v_+)b, a\psi(v^-)b)\; \text{ if }\; v^{(L)}=a \text{ and } (a\psi(v^-)b, a\psi(v_+)b) \; \text{ if }\; v^{(L)}=b. \]
\end{prop}

By Proposition~\ref{cor:Lyn+-} we have that if $v$ is not constant, then for any $x\in \Aa $
\begin{equation}\label{newone}
|a\psi(v)b|_x= |a\psi(v^-)b|_x+ |a\psi(v_+)b|_x.
\end{equation}
The following proposition is a direct consequence of~\eqref{newone}. It gives a remarkable interpretation of the pair  of words $v_+$ and $v^-$ in the combinatorics of Christoffel words.
Recall that the \emph{mediant} of the two fractions  $a/b$ and $c/d$  is the fraction $(a+c)/(b+d)$.

\begin{prop}\label{prop:+---} If  $v\in \Aa ^*$ is not constant, then
the slope of the Christoffel word $a\psi(v)b$ is the mediant of the slopes of $a\psi(v_+)b$ and $a\psi(v^-)b$.
\end{prop}

\begin{remark}
 Recall \cite{adlADL} that the slope of the Christoffel word $a\psi(v)b$ is equal to the reduced fraction $\Sb(v)$ labeling the node
(word) $v$ in the Stern-Brocot tree. From the construction of this tree $\Sb(v) = \Sb(v_1)\oplus  \Sb(v_2)$, where $\oplus$ denotes the mediant operation, and $v_1$ and $v_2$ are  the nearest ancestors of $v$ above and to the right, and above and to the left respectively. It is readily verified that $\{v_1, v_2\}= \{v_+, v^-\}$ so that in any case
 $\Sb(v) = \Sb(v_+) \oplus \Sb(v^-)$. 
\end{remark}

 The following  Propositions~\ref{lem:bar}, ~\ref{prop:+-}, and~\ref{prop:periods}  have been proved in \cite{adlADL}.

\begin{prop}\label{lem:bar} For any $v\in \Aa ^+$,
$ \pi(\psi(v\sptilde )) = |a\psi(v)b|_{{\bar v}^{(F)}}$.
\end{prop}

\begin{prop}\label{prop:+-} If  $v\in \Aa ^*$ is not constant, then
\[|a\psi(v)b|= |a\psi(v^-)b|+ |a\psi(v_+)b|=  |a\psi(^-v)b|+ |a\psi(_+v)b|.\]
Moreover, $ |a\psi(_+v)b| = |a\psi(v)b|_{{\bar v}^{(F)}}.$
\end{prop}

\begin{prop}\label{prop:periods} For any word $v=v_1\cdots v_n$, with $n>0$, $v_i\in \Aa $, $i=1,\ldots,n$, one has
\[|\psi(v)|= \sum_{i=1}^n \pi(\psi(v_1\cdots v_i))= \sum_{i=1}^n|a\psi(v_i\cdots v_n)b|_{{\bar v}_i}.\]
\end{prop}

For any $v\in \Aa ^*$ let  $\Ra(v)$ denote the ratio  $\Ra(v)= \frac{p_a(v)}{p_b(v)}$. We recall  \cite{adlADL} that the reduced fraction $\Ra(v)$ labels the node (word) $v$ in the Raney tree. The following remarkable proposition, which is readily derived from  Propositions~\ref{can:duebis} and~\ref{prop:+-},  holds:
\begin{prop} Let $v$ be a non-constant word over $\Aa $. 
 If  $v^{(L)}=a$ (resp., $v^{(L)}=b$), then
\[\Ra(v)= \frac{|a\psi(v_+)b|}{|a\psi(v^-)b|}, \quad \left (\text{resp., } \Ra(v)= \frac{|a\psi(v^-)b|}{|a\psi(v_+)b|}\right).\]
 If  $v^{(F)}=a$ (resp., $v^{(F)}=b$), then
\[\Sb(v)= \frac{|a\psi(_+v)b|}{|a\psi(^-v)b|}, \quad \left(\text{resp., } \Sb(v)= \frac{|a\psi(^-v)b|}{|a\psi(_+v)b|}\right).\]
\end{prop}

 An interesting interpretation of the extension $\ext(v)$ of a directive word $v$
 of the central word $\psi(v)$ is given by the following:
  
  \begin{prop} Let $v= v_1v_2\cdots v_m$, $v_i\in \Aa $, $i=1,\ldots,m$, be a word of $\Aa ^+$. Let
  $w= \psi(v)= w_1\cdots w_k$ with $k= |\psi(v)|$ and $w_i\in \Aa $, $i=1,\ldots,k$.  Then one has
\[ \ext(v)= \card \{\pi(\psi(v_1\cdots v_i)) \mid  1\leq i \leq m\} = \card\{ \pi(w_1\cdots w_i) \mid 1\leq i \leq k\} .\]
  \end{prop}
  \begin{proof} Let  $i=1,\ldots,m$ and set $u= v_1\cdots v_i$. For $x\in \Aa $ one has by~\eqref{eq:minimalperiod}
  $\pi(\psi(ux))=  p_x(u)$. If $x=u^{(L)}=v_i$, then by~\eqref{eq:perpsi}, $p_{v_i}(u)= \pi(\psi(u))$ and the minimal period is unchanged.
  If $x={\bar u}^{(L)}= {\bar v}_i$, then $p_{{\bar v} _i}(u) > \pi(\psi(u))$. Hence, if $v= x_0^{\alpha_0}\cdots  x_n^{\alpha_n}$, the set
  of  distinct minimal periods of $\psi(v_1\cdots v_i)$,  $i=1,\ldots,m$, is formed by the minimal periods of the words
\[ \psi(x_0), \psi(x_0^{\alpha_0}x_1), \ldots, \psi(x_0^{\alpha_0}\cdots  x_{n-1}^{\alpha_{n-1}}x_n)\]
  whose number is $n+1= \ext(v)$.
  
  Now let  $w_1w_2 \cdots w_r$ with $r \leq k$ be a non-empty prefix of $w$. There exists  $1\leq i < m$ such that
\[\psi(v_1\cdots v_i)v_{i+1} \leq_p   w_1w_2 \cdots w_r \leq_p \psi(v_1\cdots v_{i+1}),\]
  where we let $\leq_p$ denote the prefixal ordering. 
  Hence, $\pi(\psi(v_1\cdots v_i)v_{i+1}) \leq \pi(w_1w_2 \cdots w_r) \leq \pi( \psi(v_1\cdots v_{i+1})) $.  By 
~\eqref{eq:minimalperiod},
\[\pi(\psi(v_1\cdots v_i)v_{i+1})=\pi( \psi(v_1\cdots v_{i+1}))= \pi(w_1w_2 \cdots w_r).\]
 Thus between  $\pi(\psi(v_1\cdots v_i))$ and $\pi( \psi(v_1\cdots v_{i+1})) $ there are no new minimal periods. From this
  the result follows.
  \end{proof}
  \begin{cor} For each $k>0$ and $v\in \Aa ^k$ the word $w= \psi(v)$  has the maximum number of distinct minimal periods of its prefixes if and only if $v$ is alternating, i.e., $w$ is a palindromic prefix of $f$ or of $E(f)$.
  \end{cor}
  \begin{proof} By the previous proposition the number of distinct minimal periods of $w=\psi(v)$ is given by $\ext(v)$. 
  A word  $v\in \Aa ^k$ attains the maximum value $k$ of $\ext(v)$ if and only if $v$ is alternating.
  \end{proof}
  If $v= x_0^{\alpha_0}\cdots  x_n^{\alpha_n}$  we set
\[ \pi_i(v) = \pi( \psi(x_0^{\alpha_0}\cdots  x_{i-1}^{\alpha_{i-1}}x_i)),  \ 0\leq i \leq n.\]
  Moreover, we let ${\bar \pi}$ denote the  arithmetic mean of the distinct minimal periods $\pi_i$, $0\leq i \leq n$.
  
  \begin{cor} For $v\in \Aa ^+$ one has:
\[ \frac{|\psi(v)|}{\ext(v)} \geq {\bar \pi},\]
  where the equality holds if and only if $v$ is alternating.
  \end{cor}
  \begin{proof} Let $n+1=\ext(v)$. By Proposition~\ref{prop:periods} one has
\[|\psi(v)|= \sum_{i=1}^{|v|} \pi(\psi(v_1\cdots v_i))= \sum_{i=0}^{n}\alpha_i\pi_i \geq \sum_{i=0}^{n}\pi_i,\]
  so that dividing for $n+1$ we have
\[  \frac{|\psi(v)|}{n+1} \geq  \frac{\sum_{i=0}^{n}\pi_i}{n+1}= {\bar \pi}.\]
   The equality holds if and only if $\alpha_i=1$, $i=0, \ldots, n$. From this the result follows.
  \end{proof}

\section {Christoffel morphisms}\label{sec:cmor}
Let $x\in \Aa $ and $y={\bar x}$, we consider the injective endomorphism $\mu_x\sptilde $ of $\Aa ^*$ defined by  $\mu_x\sptilde (x) = x$ and $\mu_x\sptilde (y) = yx$.
In the following, we shall set
\[\lambda_a= \mu_a  \ \text {and} \  \lambda_b=  \mu_b\sptilde ,\]
and  for any  $v=v_1v_2\cdots v_n$, $v_i\in \Aa $, $1\leq i \leq n$, we define:
\[\lambda_v= \lambda_{v_1}\circ \lambda_{v_2} \circ \cdots \circ \lambda_{v_n} .\]
If $v=\varepsilon$, we set $\lambda_{\varepsilon}= \mathrm{id}$. Thus $\{\lambda_a, \lambda_b\}^*= \{\lambda_v \mid v\in \Aa ^*\}$.

The following lemma shows that the morphism $\lambda_b$ is \emph{right conjugate} \cite[Sect.~2.3.4]{LO2}  to $\mu_b$.

\begin{lemma}\label{lem:aux1} For any $v\in \Aa ^*$, \ 
$ b\lambda_b(v)= \mu_b(v)b.$
\end{lemma}
\begin{proof} By induction on the length of $v$. The result is trivially verified if $|v|\leq 1$. Let us then suppose $|v|>1$ and write $v= ux$ with $x\in \Aa $. If $x=a$ then, by using the inductive hypothesis,
\[ b\lambda_b(ua)= b\lambda_b(u)\lambda_b(a)= \mu_b(u)bab = \mu_b(ua)b.\]
If $x=b$, one has:
\[ b\lambda_b(ub)= b\lambda_b(u)b= \mu_b(u)bb= \mu_b(ub)b. \qedhere\]
\end{proof}
\begin{prop}\label{can:uno} For all $v\in \Aa ^*$,
\[\lambda_v(ab) = a\psi(v)b. \]
\end{prop}
\begin{proof} By induction on the length of $v$. If $|v|\leq 1$, the result  is trivially verified. Suppose $|v|>1$ and write $v=xw$ with $x\in \Aa $ and $w\in \Aa ^*$. By induction one has:
\[ \lambda_{xw}(ab)= \lambda_x(\lambda_w(ab))=  \lambda_x(a\psi(w)b).\]
Let us first suppose that  $x=a$. In such a case $\lambda_a= \mu_a$. By Justin's formula
\[ \lambda_{aw}(ab)= \mu_a(a\psi(w)b) = a\mu_a(\psi(w))ab= a\psi(aw)b.\]
Let now $x=b$, so that $\lambda_b= \mu_b\sptilde $. By Lemma~\ref{lem:aux1} and Justin's formula one has:
\[ \lambda_{bw}(ab)= ab\lambda_b(\psi(w))b = a\mu_b(\psi(w))bb = a\psi(bw)b. \qedhere\]
\end{proof}
\begin{cor}\label{prop:DER} For any  $w, v\in \Aa ^*$,  
\[a\psi(wv)b = \lambda_w(a\psi(v)b) .\]
\end{cor}
\begin{proof} By the preceding proposition one has:
\[a\psi(wv)b= \lambda_{wv}(ab)= \lambda_w(\lambda_v(ab))= \lambda_w(a\psi(v)b). \qedhere\]
\end{proof}
\begin{prop}\label{can:due} Let $v\in\Aa ^*$ be non-constant. The following holds:
\[ \lambda_v(a)= \lambda_{v_+}(ab)= a\psi(v_+)b,  \  \lambda_v(b)= \lambda_{v^-}(ab)= a\psi(v^-)b, \text{ if } v^{(L)}=a \]
and
\[ \lambda_v(a)= \lambda_{v^-}(ab)= a\psi(v^-)b, \  \lambda_v(b)= \lambda_{v_+}(ab)= a\psi(v_+)b, \text{ if } v^{(L)}=b. \]
\end{prop}
\begin{proof}
We shall prove the result only when $v^{(L)}=a$. The case $v^{(L)}=b$ is similarly dealt with. We can write
$v= v_+ba^r$ for a suitable $r>0$. Therefore, by Proposition~\ref{can:uno} one has:
\begin{equation}\label{eq:conc1}
  \lambda_v(a)=   \lambda_{v_+ba^r}(a) =  \lambda_{v_+b}(a)=  \lambda_{v_+}(ab) = a\psi(v_+)b,
  \end{equation}
and
\begin{equation}\label{eq:conc2}
 \lambda_v(b)= \lambda_{v^-a}(b)= \lambda_{v^-}(ab)=  a\psi(v^-)b,  
\end{equation}
which proves the assertion.
\end{proof}
It is worth noting the similarity existing between Proposition~\ref{prop:J0000}, Corollary~\ref{cor:J0000C}, and Proposition~\ref{can:duebis}    concerning standard words and 
the morphisms $\mu_v$, $v\in \Aa ^*$, which preserve standard words \cite{adl97}, and Proposition~\ref{can:uno}, Corollary~\ref{prop:DER}, and Proposition~\ref{can:due} concerning Christoffel words and the morphisms $\lambda_v$, $v\in \Aa ^*$, which, as we shall see soon (cf. Proposition~\ref{CH:morph}), preserve Christoffel words.

Let $\mathcal{M}_{CH}$ denote the monoid of all endomorphisms $f$ of  $\Aa ^*$  which preserve Christoffel words, i.e., if  $w\in \CH $,
then $f(w)\in \CH $. Such a morphism $f$  will be called \emph{Christoffel morphism}. The following proposition was proved in \cite{BLRS} by a different (geometrical) technique\footnote{In \cite{BLRS} any Sturmian morphism, i.e., any endomorphism of ${\cal A}^*$ which preserves Sturmian words, is called Christoffel morphism.}.
\begin{prop}\label{CH:morph}
$\mathcal{M}_{CH}= \{\lambda_a, \lambda_b\}^*$.
\end{prop}
\begin{proof} Let  $\lambda_v \in \{\lambda_a, \lambda_b\}^*$ and $w\in \CH $. We prove that  $\lambda_v (w)\in \CH $.  Let us first suppose that $w$ is a proper Christoffel word. We can write $w= a\psi(u)b$ for a suitable $u\in \Aa ^*$.  Thus by Proposition~\ref{can:uno}
\[ \lambda_v(w)= \lambda_v(a\psi(u)b)= \lambda_v(\lambda_u(ab))= \lambda_{vu}(ab) = a\psi(vu)b \in \CH .\]
Let us now suppose that $w\in \Aa $. Let  $w=a$. If $v$ is not constant, then the result follows from Proposition~\ref{can:due}. Let us suppose that
$v$ is constant. The result is trivial if $v=\varepsilon$. If $v=a^k$ with $k>0$, we have $\lambda_{a^k}(a)=a \in \CH $. If $v=b^k$ one has
$\lambda_{b^k}(a)= ab^k\in \CH $. In a similar way one proves the result if $w=b$.

Let now $f$ be any Christoffel morphism. Since
\[f(a), f(b), f(ab)=f(a)f(b)\in \CH ,\] one has that  $(f(a),f(b))$ is the standard factorization of
$f(ab)$ in Christoffel (Lyndon) words. As $f(ab)$ is a proper Christoffel word we can write $f(ab)= a\psi(v)b$. Let us suppose that $v$ is not constant. If $v^{(L)}=a$, by Proposition~\ref{cor:Lyn+-} one has that the standard factorization of $a\psi(v)b$ in Lyndon words is $(a\psi(v_+)b, a\psi(v^-)b)$.
This implies, in view of~\eqref{eq:conc1} and~\eqref{eq:conc2},
\[f(a)= a\psi(v_+)b = \lambda_{v_+}(ab) = \lambda_v(a), \ f(b)= a\psi(v^-)b= \lambda_{v^-}(ab)= \lambda_v(b). \]
Hence, in this case $f= \lambda_v$ and the result follows. The case $v^{(L)}=b$ is similarly dealt with. 

Let us now suppose that $v$ is  constant.  We suppose that $v=a^k$.  One has
$f(ab)= a\psi(a^k)b = \lambda_{a^k}(ab)= aa^kb$. In this case $f(a)= a = \lambda_{a^k}(a)$ and $f(b)= a^kb = \lambda_{a^k}(b)$. Hence,
$f= \lambda_{a^k}$. In a similar way if $v=b^k$ one obtains $f= \lambda_{b^k}$. Thus the result is completely proved. 
\end{proof}
\begin{prop}\label{prop:preimage} Let $v, w\in \Aa ^*$.  If $\lambda_v(w)\in \CH $, then $w\in \CH $.
\end{prop}
\begin{proof} Let us first prove that for $x\in \Aa $, if $\lambda_x(w)\in \CH $ then $w\in \CH $. If $\lambda_x(w)=y\in \Aa $, then the only possibility is $x=y$ and $w=x\in \CH $. Let us then suppose that $\lambda_x(w)$ is a proper Christoffel word $a\psi(v)b$ for a suitable $v\in \Aa ^*$. We can write in view of Proposition \ref{can:uno}
\begin{equation}\label{eq:conc3}
 \lambda_x(w)= a\psi(v)b= \lambda_v(ab).
 \end{equation}
If $v=\varepsilon$, then one obtains $ \lambda_x(w)=ab$ and $w\in \CH $. Let us then suppose $|v|>0$. We wish to prove that $v^{(F)}=x$. To this end we show that $x=a$ if and only if $v^{(F)}=a$. Indeed, as $\lambda_a(w)\in \{a, ab\}^*$ and $\lambda_b(w)\in \{b, ab\}^*$ if $v^{(F)}=a$,
as $\psi(v)$ begins with $a$,  it follows that $x=a$. Conversely, suppose that $x=a$; one has that $w$ has to terminate with $b$. Moreover,
if $w=b^n$, with $n>1$  one would have $\lambda_a(b^n)= (ab)^n\not\in \CH $. If $n=1$, then $\lambda_a(b)= ab$ and $v=\varepsilon$, a contradiction. Hence, in $w$ there must be at least one occurrence of the letter $a$, so that we can write $w=w'ab^r$ with $r>0$. 
Thus
$ab\in \Ff(w)$. This implies that $aab\in \Ff(\lambda_a(w))$, so that $v^{(F)}=a$. We have then proved that $v^{(F)}=x$.  Writing $v=xv'$ from~\eqref{eq:conc3} we
have
\[  \lambda_x(w)= a\psi(v)b= \lambda_{xv'}(ab) = \lambda_x(\lambda_{v'}(ab)).\]
As $\lambda_x$ is injective, it follows $w= \lambda_{v'}(ab)\in \CH $.

The remaining part of the proof is obtained by induction on the length of $v$. If $|v|>1$,
set $v= xv'$ and suppose that $\lambda_v(w)\in \CH $. We can write $\lambda_v(w) = \lambda_x(\lambda_{v'}(w))\in \CH $. It follows from we have previously proved 
that $\lambda_{v'}(w)\in \CH $ and by induction $w\in \CH $.
\end{proof}

The following lemma relates the morphisms  $\lambda_{a^kb}$ and $\mu_{a^kb}$, $k\geq 0$.

\begin{lemma}\label{lemma:prepl} For each $k\geq 0$ and $v\in \Aa ^*$,
\[ \lambda_{a^kb}(bv)= \mu_{a^kb}(vb), \quad \lambda_{a^kb}(av) = a \mu_{a^kb}(vb).\]
\end{lemma}
\begin{proof} By Lemma~\ref{lem:aux1} one has  $\lambda_{a^kb}(bv)= \lambda_{a^k}(b\lambda_b(v)) =  \mu_{a^k}(\mu_b(v)b)=
\mu_{a^kb}(vb)$ and the first equation is proved. By using again Lemma~\ref{lem:aux1} one has
$\lambda_{a^kb}(av)= \lambda_{a^k}(ab\lambda_b(v))=  \lambda_{a^k}(a\mu_b(v)b))= \mu_{a^k}(a\mu_b(vb)) = a \mu_{a^kb}(vb).$
\end{proof}

\begin{lemma} For all $v\in \Aa ^*$ and $x\in \Aa $,
\[ |\lambda_v(x)| = |\mu_v(x)|= p_x(v) .\]
\end{lemma}
\begin{proof} Let us first suppose that $v$ is constant.  If $v=a^k$ with $k\geq 0$, then since $\lambda_{a^k}= \mu_{a^k}$, in view of~\eqref{eq:muvi}  the result trivially follows. If  $v=b^k$, then $|\lambda_{b^k}(a)|= |ab^k|=  |b^ka| = |\mu_{b^k}(a)|$ and $|\lambda_{b^k}(b)|= |\mu_{b^k}(b)| =1$
and the result is achieved. If $v$ is not constant the result follows from Propositions~\ref{can:due} and~\ref{can:duebis}.
\end{proof}

\section{Derivative of a Christoffel word}\label{sec:derc}

Let $\varphi: \Aa ^* \rightarrow \Aa ^*$ be an injective morphism. As is well-known (cf.~\cite{codes}) the set $X= \varphi(\Aa )$ is a code over the alphabet $\Aa $, i.e., any word of $X^+$ can be uniquely factored by the elements of $X$. Thus there exists an isomorphism, that we still denote
by $\varphi$, of $\Aa ^*$ and $X^*$. Let $\varphi^{-1}$ be  the inverse morphism of $\varphi$.

If $w\in X^+$,  $\varphi^{-1}(w)$ is a uniquely determined word over the alphabet $\Aa $, that we call
\emph{derivative} of $w$ with respect to $\varphi$. We shall denote $\varphi^{-1}(w)$ by $\partial_{\varphi} w$, or simply
$\partial w $, when there is no ambiguity.
\begin{example}
Let $X= \{ab,ba\}$ and $\varphi$ the Thue-Morse morphism defined by $\varphi(a)=ab$
and $\varphi(b)=ba$. One has that $\partial abbabaab = abba$.

 Let  $w$ be the finite Sturmian word $w= aababaaba$.  The word $w$
can be decoded by the morphism $\mu_a :\{a, b\}^* \rightarrow  \{a , ab\}^*$ or by the morphism $\mu\sptilde _a :\{a, b\}^* \rightarrow  \{a , ba\}^*$. In the first case one obtains the derivative $w_1 = abbaba$ which is still a finite Sturmian word, whereas in the second case one gets the derivative $w_2= aabbab$ which is not a finite Sturmian word. 
\end{example}

In the study of derivatives of finite words over $\Aa $ belonging to a given class $\mathcal{C}$,  we require that the set $\mathcal{M}$ of injective endomorphisms of $\Aa ^*$ satisfies the two following basic conditions:
\begin{enumerate}[1.]
\item If $\varphi \in \mathcal{M}$, then for any $w\in \mathcal{C}$, $ \varphi(w)\in \mathcal{C}.$
\item If $\varphi(v)= w$ and $w\in \mathcal{C}$, then $v\in \mathcal{C}$.
\end{enumerate}
Moreover, one can restrict the class $\mathcal{M}$ of endomorphisms  to some subclass $\hat{\mathcal{M}}$ assuring that the
obtained derivatives satisfy suitable combinatorial properties.

In this section we shall consider the class $\mathcal{C}$ of  Christoffel words. We define a derivative of a proper Christoffel word   by referring
to a suitable  Christoffel morphism.  A derivative 
 in the case of finite (and infinite) standard Sturmian words and its relation with the previous one will be  given  in Section~\ref{sec:ders}.

Let  $u= \psi(v)$ be a central word.
We define \emph{index} of the central word $u$ the integer $0$ if $v=\varepsilon$ or, otherwise, the first 
 element in the integral representation $(\alpha_0, \alpha_1, \ldots, \alpha_n)$  of $v$, i.e., 
 $\alpha_0$. We let $\ind(u)$ denote the index of $u$. If $w= aub$ is a proper Christoffel word
we define \emph{index} (resp., \emph{directive word}) of $w$ the index (resp., directive word) of the central word $u$.

In the following, for $x\in \Aa $, we set  $\PER_x= \PER \cap \ x\Aa ^*$ and for any $k\geq 0$
we define the prefix code $X_k$ and the suffix code $Y_k$:
\begin{equation}\label{eq:xy}
 X_k=  \{a^kb, a^{k+1}b\} \ \text{ and }\ Y_k= \{ab^k, ab^{k+1}\}.
 \end{equation}

\begin{lemma}\label{lem:AABC} Let $w= aub$ be a proper Christoffel word with $u\neq \varepsilon$ and $k$ be the index of $u$. If  $u\in \PER_a$, then $w\in a^{k+1}bX_k^*$. If $u\in \PER_b$, then
$w\in ab^kY_k^*$.
\end{lemma}
\begin{proof} We shall prove the lemma only in the case $u\in \PER_a$. The case  $u\in \PER_b$ is dealt with in a similar way. We shall denote by $Z_k$ the set $a^{k+1}bX_k^*$. The proof is by induction on the length of the directive word $v$ 
of the central word $u= \psi(v)$ of index $k$. If $v=a^k$ then $u = a^k$ and $w=aub= a^{k+1}b$.
If $v= a^kb$, then $u= a^kba^k$ and $w= a^{k+1}ba^kb$ and we are done. Let us then suppose
that the result is true for a directive word $v\in a^kb\Aa ^*$ and prove it for the directive word
$vx$ with $x\in \Aa $. Since $\psi(v)$ begins with $a^kba^k$, we can write from Proposition~\ref{Prop:uno}
\begin{equation}\label{eq:psi}
 u= \psi(v)= u_1bau_2= u_2abu_1= a^kba^k\zeta,
 \end{equation}
where $\zeta\in \Aa ^*$ and $u_1,u_2\in \PER_a$. Moreover, from Lemma~\ref{lem:uno} one
has:
\[ \psi(va)= u_1bau_2abu_1  \ \text{ and }\  \psi(vb)= u_2abu_1bau_2. \] 
Let $z_1$ and $z_2$ be the Christoffel words  $z_1= a\psi(va)b$ and $z_2= a\psi(vb)b$. One has
that 
\begin{equation}\label{eq:two} 
       z_1= au_1b au_2abu_1b= au_1b a\psi(v)b = au_1b w \text{ and }z_2= wau_2b.
 \end{equation}
      From~\eqref{eq:psi} one has that $|u_1|\geq k$. This implies that the index of $u_1=\psi(v_1)$ is $k$.
      Since $au_1b$ is a Christoffel word and $|v_1|<|v|$, one has by induction  $au_1b \in Z_k$. Also by induction
      $w\in Z_k$. Hence, by~\eqref{eq:two} one has $z_1\in Z_k$.
      
     As regards $z_2$  from~\eqref{eq:psi} one has either $|u_2|\geq k$ or   $|u_2|= k-1$. In the first case since $\ind(u_2)= k$,  in a way similar as above one derives by induction that the Christoffel word $au_2b \in Z_k$, that implies by~\eqref{eq:two}, as $w\in Z_k$, that $z_2\in Z_k$.
    In the second case $au_2b= a^kb\in X_k$, so that, as $w\in Z_k$, it follows $z_2\in Z_k$ and this concludes the proof.
\end{proof}
If $w$ is a proper Christoffel word, we can introduce a derivative of $w$ as follows. If $w=aub$,
where  $u\in \PER_a$ is a central word of index $k$, we consider the prefix code $X_k$ and
the injective endomorphism $\varphi_k$ of $\Aa ^*$ defined by 
\begin{equation}\label{eq:fi1}
 \varphi_k(a) = a^{k+1}b, \quad \varphi_k(b) = a^{k}b.
\end{equation}
By the previous lemma $w\in X_k^*$ and the derivative of $w$ with respect to $\varphi_k$
is  $\partial_k w = \varphi_k^{-1}(w)$. Let us observe that from the definition for all $k\geq 0$ one has $\partial_k a^{k+1}b =a$
whereas $\partial_{k+1} a^{k+1}b =b$.

 In the case $u\in \PER_b$, 
one can consider  the injective endomorphism $\hat\varphi_k$ of $\Aa ^*$ defined by 
\begin{equation}\label{eq:fi2}
\hat \varphi_k(a) = ab^k, \quad \hat \varphi_k(b) = ab^{k+1}.
\end{equation}
By the previous lemma $w\in Y_k^*$ and the derivative of $w$ with respect to $\hat\varphi_k$
is  $\hat\partial_k w = \hat\varphi_k^{-1}(w)$. Observe that  for all $k\geq 0$ one has $\hat\partial_k ab^{k+1}=b$
whereas $\hat\partial_{k+1} ab^{k+1}=a$.

If $w= aub$ is a proper Christoffel word of index $k>0$ the \emph{derivative} of $w$ is the word
$\partial w =\partial_k w$ if $u\in \PER_a$ and $\partial w =\hat\partial_k w$ if $u\in \PER_b$.
Finally, if $k=0$, i.e., $w=ab$, we set $\partial ab= a$.

Let us observe that from the definition one has for each $k\geq 0$:
\[ \varphi_k = \lambda_{a^kb} \  \text{ and }\  \hat\varphi_k = \lambda_{b^ka},\]
so that by Proposition~\ref{CH:morph},   $\varphi_k$ and   $\hat\varphi_k$  are Christoffel morphisms.

\begin{example}
\label{ex:index00}
Let $v=ab^2a^2b$ and $w$ be the Christoffel word $a\psi(v)b$ where $\psi(v)$ is a  central 
word of index 1. One has:
\[ w = aababaababaabababaababaababab.\]
In this case one has $X_1=\{a^2b, ab\}$ and
$\partial w = \partial_1w=  abababbababb.$

If $w=a\psi(b^2a^2)b$, then $w=abbabbabbb$. The index of $w$ is 2 and $Y_2=\{ab^2, ab^3\}$
and $\partial w = \hat\partial_2w=  aab$.
\end{example}
\begin{remark}
Let us explicitly observe that two different Christoffel words can have the same derivative. For instance, the Christoffel words $w= a\psi(a^2b^2a)b$ and $w'= a\psi(b^2aba)b$ have both the
derivative $ababb= a\psi(ba)b$. Moreover, from the definition it follows that all proper Christoffel words having directive words which are equal up to their index have the same derivative, i.e., for all $k>0$,  $x\in \Aa $ and $\xi\in {\bar x}\Aa ^*$, one has $\partial(a\psi(x^k\xi)b) = \partial(a\psi(x\xi)b)$. A proper Christoffel word $w$ is determined by its derivative $\partial w$ and the value of its index.
\end{remark}

\begin{thm}\label{thm:diretto} The derivative of a proper Christoffel word is a Christoffel word.
\end{thm}

\begin{proof}Let $w =aub$ be a proper Christoffel word of index $k$.   If $k=0$, i.e., $w=ab$,  $\partial ab= a$ and in this case the result is trivially verified. Let us suppose $k>0$. The  derivative of $w$ is 
$\partial w = \varphi_k^{-1} (w)$ if $u\in \PER_a$ and $\partial w =\hat\varphi_k^{-1}(w)$ if $u\in \PER_b$.
  In the first case  $\varphi_k(\partial w) =w$ and in the second  case $\hat\varphi_k(\partial w) =w$. Since $\varphi_k$ and $\hat\varphi_k$  are  Christoffel morphisms, by Proposition~\ref{prop:preimage} it follows that in both cases $\partial w \in \CH $. 
\end{proof}
\begin{cor} Let $w=a\psi(v)b$ be a Christoffel word of index $k$, with $ v$ non-constant. If $w=w_1w_2$ with $w_1,w_2\in \CH $ is the standard factorization of $w$ in Lyndon words, then $w_1, w_2\in X_k^*$ or $w_1, w_2\in Y_k^*$ and $\partial w=  \partial w_1\partial w_2$ with  $\partial w_1, \partial w_2\in \CH $, is the standard factorization of $\partial w$
in Lyndon words.
\end{cor}
\begin{proof} By Theorem~\ref{thm:diretto}, $\partial w$ is a Christoffel word. Since $v$ is not constant, by Proposition~\ref{Prop:uno} 
one has $\psi(v)= u_1bau_2=u_2abu_1$ with $u_1,u_2\in \PER$. Hence, $w= w_1w_2$ where
$w_1=au_1b$ and $w_2= au_2b$ are two proper Christoffel words. In view of Theorem ~\ref{thm:Lynd},  $w_1w_2$ is the standard factorization of $w$ in Lyndon words.

 Let us suppose that $v\in a\Aa ^*$.  One has $\psi(v)= a^kba^k\xi$, with $\xi\in \Aa ^*$ from which, as we have seen in the proof of Lemma~\ref{lem:AABC}, one  derives that $w_1, w_2\in X_k^*$.  Thus $\partial w= \partial_k w = \partial_k w_1\partial_k w_2$ with  $\partial_k w_1, \partial_k w_2\in \CH $. By Theorem~\ref{thm:Lynd} the result follows.  The case $v\in b\Aa ^*$ is similarly dealt with.
\end{proof}

\begin{thm}\label{thm:inverso} Let $k\geq 1$ and $w\in X_k^*\cup Y_k^*$. If $\partial w$ is a Christoffel word, then $w$ is a proper Christoffel word.
\end{thm}
\begin{proof}We shall suppose that $w\in X_k^*$. A similar proof can be done when $w\in Y_k^*$. One has that 
$w= \varphi_k(\partial w)$. Since  $\varphi_k$ is a Christoffel morphism, it follows that $w\in CH$. Moreover, it is readily verified that $w$ is a proper Christoffel word.
\end{proof}

From Theorems~\ref{thm:diretto} and~\ref{thm:inverso} it follows:

\begin{cor} Let $k>0$ and $w\in X_k^*\cup Y_k^*$. Then $w$ is a proper Christoffel word if and only if $\partial w$ is a Christoffel word.
\end{cor}

\begin{prop}\label{prop:periods1} If $w=a\psi(v)b$, then 
\[|\partial w| = \pi(\psi(v\sptilde ))=|a\psi(v)b|_{{\bar v}^{(F)}}.\]
\end{prop}
\begin{proof} From the definition of derivative of $w$ one has $|\partial w| = |a\psi(v)b|_{{\bar v}^{(F)}}$, so that the result follows from
Proposition~\ref{lem:bar}.
\end{proof}

\begin{cor} A proper Christoffel word $w = a\psi(v) b$ is uniquely determined by  $v^{(F)}$, $|w|$, and $|\partial w|$.
\end{cor}
\begin{proof} Let $w= a\psi(v)b$. By Proposition~\ref{prop:periods1}, $|\partial w| = |a\psi(v)b|_{{\bar v}^{(F)}}$, so that
$ |a\psi(v)b|_{{ v}^{(F)}} = |w| -|\partial w|$.  The  Christoffel word $w$ is uniquely determined by its  slope $\eta(w)= |w|_b/|w|_a$.
If $v^{(F)}=a$, then $\eta(w) = |\partial w|/(|w| -|\partial w|)$. If $v^{(F)}=b$, then $\eta(w) = (|w| -|\partial w|)/ |\partial w|$. From this the result follows.
\end{proof}

From Propositions~\ref{prop:periods} and~\ref{prop:periods1} one derives:

\begin{cor}\label{cor:sumderiv}For any word $v=v_1\cdots v_n$, with $n\geq 0$, $v_i\in \Aa $, $i=1,\ldots,n$, one has
\[|\psi(v)|=  \sum_{i=1}^n|\partial a\psi(v_i\cdots v_n)b|.\]
\end{cor}

\begin{example}
Let  $w= a\psi(v)b$ with $v=a^2b^2a$. One has
\[w= a^3ba^2ba^3ba^2ba^2b.\] Moreover,  $\partial a\psi(ab^2a)b= \partial w = ababb$, $\partial a\psi(b^2a)b= \partial a\psi(ba)b=ab$, and $\partial a\psi(a)b= a$. Hence, $|\psi(v)|= 15= 2\cdot 5 + 2\cdot 2+ 1$.
\end{example}

The following noteworthy  theorem relates, through their directive words, the central word of a proper Christoffel word and the central word of its derivative.

\begin{thm}\label{thm: derpart}If $w= a\psi(v)b$ and $v$ is not constant, then
\[\partial w = a\psi(_+v)b .\]
\end{thm}
\begin{proof}Let $v= x_0^{\alpha_0}\cdots  x_n^{\alpha_n}$ with $n\geq 1$ as $v$ is not constant.  We can write, setting $\alpha_0=k$ and $\alpha_1=h$,  $v=x_0^kx_1^h\xi$.  We shall first suppose that $\psi(v)\in \PER_a$, so that 
  $v=a^kb^h\xi$.
   One has $a\psi(v)b \in X_k^*$. Since $\varphi_k = \lambda_{a^kb}$ and
$(a^kb)^{-1}v = b^{h-1}\xi$, 
by Corollary~\ref{prop:DER} one obtains 
\[ a\psi(v)b = \varphi_k(\partial a\psi(v)b) = \varphi_k(a\psi(b^{h-1}\xi)b).\]
From the injectivity of $\varphi_k$  it follows  $\partial a\psi(v)b = a\psi(b^{h-1}\xi)b = a\psi(_+v)b $.

 Let us now suppose that  $\psi(v)\in \PER_b$. We can write,
 $v= b^ka^h\xi$.  One has $a\psi(v)b \in Y_k^*$. Since $\hat\varphi_k = \lambda_{b^ka}$ and
$(b^ka)^{-1}v = a^{h-1}\xi$,
by Corollary~\ref{prop:DER} one obtains 
\[ a\psi(v)b = \hat\varphi_k(\partial a\psi(v)b) = \hat\varphi_k(a\psi(a^{h-1}\xi)b).\]
From the injectivity of $\hat\varphi_k$  it follows  $\partial a\psi(v)b = a\psi(a^{h-1}\xi)b = a\psi(_+v)b$.
\end{proof}

A different  proof of Theorem~\ref{thm: derpart} based on continued fractions will be given at the end of the  section.

 \begin{example}
Let $w=aub$ with $u=\psi(a^2b^2a)$. One has
\[w= aaabaabaaabaabaab,\] $_+v= ba$, and
$\partial w = ababb= a\psi(ba)b$. If $w= a\psi(ba^2b^2a)b$, one has $_+v= ab^2a$ and  $\partial w = a\psi(ab^2a)b$. If $w=a\psi(abab)b$, then $_+v= ab$ and  $\partial w = a\psi(ab)b$.
\end{example}

\begin{cor} \label{cor:aut} Let $w$ be a Christoffel word $a\psi(v)b$ having the derivative $\partial w = a\psi(_+v)b$. Then
$\partial a\psi(E(v))b = a\psi(E(_+v))b$ and $\partial a\psi(v\sptilde )b = a\psi((v_+)\sptilde )b$. 
\end{cor}
\begin{proof} The word $v$ is not constant so that  by the previous theorem and~\eqref{eq:E+},  one has $\partial a\psi(E(v))b = a\psi(_+E(v))b= a\psi(E(_+v))b$ and  $\partial a\psi(v\sptilde )b =  a\psi(_+(v\sptilde ))b = a\psi((v_+)\sptilde )b$. 
\end{proof}
\begin{prop}\label{prop:series} Let   $v= x_0^{\alpha_0}\cdots  x_n^{\alpha_n}$. One has:
\[|a\psi(v)b|= \sum_{i=0}^{n-1}\alpha_i|a\psi( x_{i+1}^{\alpha_{i+1}-1}x_{i+2}^{\alpha_{i+2}}\cdots  x_n^{\alpha_n}) b| +\alpha_n+2.\]
\end{prop}
\begin{proof} Let $m=|v|$. By Corollary~\ref{cor:sumderiv}, 
\[ |\psi(v)|=  \sum_{i=1}^{m}|\partial a\psi(v_i\cdots v_m)b|.\]
For any $0\leq i \leq n-1$,
\[
\begin{split}
\partial a\psi( x_{i}^{\alpha_{i}}x_{i+1}^{\alpha_{i+1}}\cdots  x_n^{\alpha_n})b
&= \partial a\psi( x_{i}^{\alpha_i-1}x_{i+1}^{\alpha_{i+1}}\cdots  x_n^{\alpha_n})b=\\
\cdots &= \partial a\psi(x_{i}x_{i+1}^{\alpha_{i+1}}\cdots  x_n^{\alpha_n})b.
\end{split}
\]
Moreover,  $\partial a\psi(x_n^{\alpha_n})b = \partial a\psi(x_n^{\alpha_n-1})b = \cdots=  \partial a\psi(x_n)b$ 
and  $|\partial a\psi(x_n)b|=1$.
Hence, one has:
\[|\psi(v)|= \sum_{i=0}^{n-1}\alpha_i|\partial a\psi( x_ix_{i+1}^{\alpha_{i+1}}x_{i+2}^{\alpha_{i+2}}\cdots  x_n^{\alpha_n}) b| +\alpha_n.\]
By Theorem~\ref{thm: derpart} for all $0\leq i \leq n-1$,
\[\partial a\psi( x_ix_{i+1}^{\alpha_{i+1}}x_{i+2}^{\alpha_{i+2}}\cdots  x_n^{\alpha_n}) b = a\psi( x_{i+1}^{\alpha_{i+1}-1}x_{i+2}^{\alpha_{i+2}}\cdots  x_n^{\alpha_n}) b,\]
from which the result follows.
\end{proof}

\begin{prop} Let  $k\geq 0$ and $f$ be the function  which maps any $v\in \Aa ^k$ into $\partial a\psi(v)b$. For $v,v'\in \Aa ^k$ with $v\neq v'$ if $f(v)=f(v')$, then $v$ and $v'$ are not constant,  $_+v =  {_+v'}$, and $v= x^ry(_+v), v'= y^rx(_+v)$ with $r>0$ and $\{x,y\}= \{a, b\}$.  As a consequence the restrictions of $f$ to $a\Aa ^{k-1}$ and to $b\Aa ^{k-1}$ are injective.
\end{prop}
\begin{proof} Suppose $f(v)=f(v')$ with $v\neq v'$. If $v$ is a constant, say $v=a^k$, then $f(v)= \partial a^{k+1}b =a$. As it is readily verified
for no other word $v'$ of $\Aa ^k$ one can have $f(v')=f(v)=a$ which is a contradiction.  Since both $v$ and $v'$ are not constant, one has:
\[ \partial a\psi(v)b = a\psi(_+v)b= \partial a\psi(v')b = a\psi(_+v')b.\]
Hence, $\psi(_+v)= \psi(_+v')$. Since $\psi$ is injective, it follows $_+v = {_+v'}$. Since $v$ and $v'$ have the same length and $v\neq v'$,
  $v= x^ry(_+v), v'= y^rx(_+v)$ with $r>0$ and $\{x,y\}= \{a, b\}$. The remaining part of the proof trivially follows.
\end{proof}

The following important and well-known theorem concerning the slope of a proper Chri\-stoffel word holds (cf.\cite{BDL}):
\begin{thm}\label{thm:cf} Let $w=aub$ be a proper Christoffel word with $u=\psi(v)$ and $(\alpha_0,\alpha_1,\ldots, \alpha_n)$ be
the integral representation of  $v$. Then the slope of $w$ is given by the continued fraction
\[[\alpha_0;\alpha_1, \ldots, \alpha_{n-1}, \alpha_n+1] \text{ if }v^{(F)}=b \]
and
\[[0; \alpha_0, \alpha_1, \ldots, \alpha_{n-1}, \alpha_n+1] \text{ if }v^{(F)}=a . \]
\end{thm}

\begin{example}
Let $v=a^2b^2a$. One has $w= a^3ba^2ba^3ba^2ba^2b$ and $\eta(w)=[0;2,2,2]$=$\frac{5}{12}$. If $v=ba^2b$, then $w= abababbababb$ and $\eta(w)=[1;2,2]$=$\frac{7}{5}$. 
\end{example}

As a consequence of Theorems~\ref{thm: derpart} and~\ref{thm:cf} one obtains:
\begin{cor} Let $w=aub$ be a proper Christoffel word with $u=\psi(v)$ and $(\alpha_0,\alpha_1,\ldots, \alpha_n)$
the  integral representation of  $v$. The slope of $\partial w$ is given by the continued fraction
\[ [\alpha_1-1;\alpha_2,\ldots,\alpha_n+1] \text{ if }v^{(F)}=a\]
and
\[ [0; \alpha_1-1, \alpha_2, \ldots, \alpha_n+1] \text{ if }v^{(F)}=b.\]
\end{cor}

We remark that the slope of a Christoffel word $w= a\psi(v)b$ determines uniquely the directive word
$v$ of $\psi(v)$ and then $w$. Now we  can give a different  proof of Theorem~\ref{thm: derpart} by using continued fractions and 
Theorem~\ref{thm:cf}.

\begin{proof}[Second proof of Theorem~\ref{thm: derpart}]
We shall suppose that $\psi(v)\in \PER_a$ and $\alpha_0=\ind(v).$
 The case $\psi(v)\in \PER_b$ is similarly dealt with.  From the construction of the derivative of
$w$ one has:
\[ |\partial w|_a(\alpha_0+1) +|\partial w|_b \alpha_0 = |w|_a,\]
and
\[ |\partial w|_b +|\partial w|_a = |w|_b.\]
From these relations one easily obtains:
\[ \frac{1}{\eta(w)}= \alpha_0 + \frac{1}{1+ \eta(\partial w)}.\]
Let $\eta(w) = [0;\alpha_0,\ldots, \alpha_n+1]$. One derives from the previous equation:
\[[0;\alpha_1,\ldots, \alpha_n+1]= \frac{1}{1+ \eta(\partial w)},\]
from which one obtains:
\[\eta(\partial w) = [\alpha_1-1;\alpha_2,\ldots,\alpha_n+1].\]
Therefore, one has $\partial w= a\psi(v')b$ where $ v'$ has the integral representation 
$(\alpha_1-1, \alpha_2, \ldots, \alpha_n)$ and therefore is equal to $_+v$, which proves the assertion.
\end{proof}

\section{Depth of a Christoffel word}\label{sec:depth}

 From Theorem~\ref{thm:diretto} any  proper Christoffel word $w$ has a derivative $w'= \partial w$ which is still a Christoffel word. Therefore, if $w'$ is proper one can consider $\partial w' \in \CH $  that we shall denote $\partial^2 w$.
In general, for any $p\geq 1$, $\partial^p w$ will  denote the derivative of order $p$ of $w$. Since  $|\partial^p w| > |\partial^{p+1} w|$, there exists an integer $d$ such that $\partial^d w \in \Aa $;
we call $d$ the \emph{depth} of $w$.

\begin{example}
 Let $w$ be the Christoffel word $w=a\psi(ab^2a^2b)b$ of  Example~\ref{ex:index00}. One has
\[\partial w = abababbababb,\]
which is the Christoffel word $a\psi(_+v)b$ where $_+v= ba^2b$. The central word $\psi(_+v)$ is of order 1
and $ \partial^2 w = aabab = a\psi(ab)b.$
Moreover, one has $\partial^3 w = ab$ and $\partial^4 w= a$, so that the depth of $w$ is  $4$.
\end{example}

As we have previously seen, if  $v \in \Aa ^*$ is not constant,  $_+v$ is  the longest   suffix of $v$ which is immediately
 preceded  by the complementary of the first letter of $v$. If  $_+v$ is not constant one can consider
  $_+(_+v)$ and so on. Thus for any $v\in\Aa ^*$  we can define inductively  $v_{(1)}= v$ and, if  $v_{(n)}$ is not constant and $n\geq 1$,
\[v_{(n+1)} =  { _+(v_{(n)})}.\]  Since $|v_{(n+1)}|<|v_{(n)}|$, there exists an integer  $h=h(v)$ called \emph{height} of $v$, such
   that $v_{(h)}$ is  constant.  For instance, if $v=a^2b^2a$, one has $v_{(1)}= a^2b^2a$, $v_{(2)}= ba$, and
   $v_{(3)}= \varepsilon$. Hence, $h(a^2b^2a)= 3$.
   
   \begin{prop} \label{prop:height} Let $w= a\psi(v)b$ be a proper Christoffel word. The depth of $w$ is equal to the height of $v$.
   \end{prop}
   \begin{proof} If $v$ is constant, then $h=h(v)=1$ and $\partial w \in \Aa $, so that the depth of $w$ is 1.  Let us then suppose that $v$
   is not a constant. This implies $h(v)>1$. From  Theorem~\ref{thm: derpart} one derives that for  $n\leq h -1$
\[ \partial^n w = a\psi(v_{(n+1)})b.\]
   Since $v_{(h)}$ is constant, it follows that  $\partial^h  w \in \Aa $, so that the depth of $w$ is $h$.
   \end{proof}

 Let   $v= x_0^{\alpha_0}\cdots  x_n^{\alpha_n}$.  For $i \in \{0, \ldots, n\}$ we define a map
\[\delta_i(v): \{0, \ldots, n\}\rightarrow \{0,1\}\] as follows: $\delta_0(v)= \delta_n(v)=1$. For $0<i<n$,  if $\alpha_i>1$ we set $\delta_i(v) = 1$. Let $\alpha_i= 1$. If $\alpha_{i-1}>1$ we set $\delta_i(v)=0$. If 
  $\alpha_{i-1}=1$, then we set $\delta_i(v)=1$ if and only if $\delta_{i-1}(v)=0$.  Let us define for any $v\in \Aa ^+$
\[\delta(v) = \sum_{i=0}^n \delta_i(v).\]
  Moreover, we set $\delta(\varepsilon)= 1$.
  
  \begin{example}
Let $v= a^2bab^2aba$. In this case $n=6$. Denoting $\delta_i(v)$ simply by $\delta_i$,  the sequence $\delta_0\delta_1\cdots \delta_n$ is given by 
  $1011011$ and $\delta(v)= 5$.
\end{example}
   
   \begin{prop} \label{prop:height1} Let $v\in \Aa ^*$. Then $h(v)=\delta(v)$.
\end{prop}
\begin{proof}
If $v=\varepsilon$ the result is trivially true.  Let $v\neq \varepsilon$. We can write  $v= x_0^{\alpha_0}x_1^{\alpha_1}\cdots  x_n^{\alpha_n}$, $\alpha_i\geq 1$, $0\leq i\leq n$. We proceed by induction on $n$. If $n=0$ then $h(v)=\delta(v)=1$. If $n=1$ then $h(v)=\delta(v) =2$.  Let $n=2$, then $v=x_0^{\alpha_0}x_1^{\alpha_1}x_2^{\alpha_2}=v_{(1)}$. There are two cases:
\begin{enumerate}[(1)]
\item $\alpha_1=1$. One has $v_{(2)}= x_2^{\alpha_2}$ and $h(v)=\delta(v)=2$;\par
\item $\alpha_1>1$. One has $v_{(2)}= x_1^{\alpha_1-1}x_2^{\alpha_2}$, $v_{(3)}=x_2^{\alpha_2-1}$ and $h(v)=\delta(v)=3$.
\end{enumerate}

Let $n>2$, then $_+v=x_1^{\alpha_1-1}x_2^{\alpha_2}\cdots x_n^{\alpha_n}=v_{(2)}$. Since by the definition of height,  $h(v)=h(_+v)+1$ and, by induction, $h(_+v)=\delta(_+v)$,  it suffices to prove that $\delta(v) =\delta(_+v)+1$. There are two possibilities:

\begin{enumerate}[(1)]
\item $\alpha_1=1$. In this case $_+v=x_2^{\alpha_2}\cdots x_n^{\alpha_n}$ and
 $\delta_i(v)=\delta_{i-2}(_+v)$ if $i\geq 2$. 
  Indeed, if  $i=2$, as $\delta_1(v)=0$, one has  $\delta_2(v)=1=\delta_0(_+v)$. From the definition of $\delta$ it follows that  $\delta_i(v)=\delta_{i-2}(_+v)$ if $i> 2$.  Hence,
\[\delta(v) = 1+\sum_{i=2}^n\delta_i(v)= 1+\sum_{i=0}^{n-2}\delta_i(_+v)= 1+\delta(_+v).\]

\item $\alpha_1>1$. In this case $_+v= x_1^{\alpha_1-1}x_2^{\alpha_2}\cdots x_n^{\alpha_n}$ and $\delta_i(v)= \delta_{i-1}(_+v)$ for $i\geq 1$.  Indeed, if  $i=1$ as $\alpha_1>1$ one has $\delta_1(v)=1=\delta_0(_+v)$. For $i=2$ if  $\alpha_2>1$ then $\delta_2(v)=\delta_1(_+v)=1$. If $\alpha_2=1$, then $\delta_2(v)=0=\delta_1(_+v)$  because $\delta_1(v)=\delta_0(_+v)=1$. From the definition of $\delta$ it follows that $\delta_i(v)= \delta_{i-1}(_+v)$ for $i > 2$. Hence,
\[\delta(v)= \sum _{i=0}^n\delta_i(v)=1+ \sum_{i=1}^n\delta_{i}(v)=1+ \sum_{i=0}^{n-1}\delta_{i}(_+v) =1+\delta(_+v). \qedhere\]
\end{enumerate} 
\end{proof}

   \begin{example}
 Let $v= a^3bab^2aba$. One has $h(v)=5$ and $\delta_0\delta_1\delta_2\delta_3\delta_4\delta_5\delta_6= 1011011$,
   so that $\delta(v)=5.$
\end{example}
   
   Let   $v= x_0^{\alpha_0}\cdots  x_n^{\alpha_n}$. If 
   $\alpha_i>1$ for all  $0 <i<n$, then from the definition of $\delta$ one has $\delta(v) = \ext (v)=n+1$. Let us suppose on the
   contrary that    $\alpha_i=1$ for all  $0 <i<n$. We can write  $v = x_0^{\alpha_0} u x_n^{\alpha_n}$ where
   $u$ is an alternating word $u=x_1x_2\cdots x_{n-1}$. In this case it is easy derive  that
\[ \delta(v) = 2 + \left\lfloor \frac{n-1}{2} \right\rfloor= \ext (v) - \left \lceil \frac{|u|}{2} \right \rceil.\]
   In general, by grouping together consecutive  $x_i$, $0<i<n$, having $\alpha_i=1$ we can rewrite $v$ uniquely as
   \begin{equation}\label{eq:sim0}
    v = v_0u_1v_2u_3 \cdots u_{k-1}v_k, 
    \end{equation}
   where all terms of the integral representation of $u_i$ (resp., $v_i$), $1\leq i \leq k-1$ are equal to $1$ (resp., $>1$) and all terms
   of the integral representation of $v_0$ (resp., $v_k$) are  $>1$, with the possible exception of the first (resp., last). 
    
   We call the $u_i$, $i= 1, 3, \ldots, k-1$,  the \emph{alternating components}  of $v$. 
   For example,
   if $v= a^3b^2abab^2aba^2ba$, then we can factore it as $v= (a^3b^2)(aba)(b^2)(ab)(a^2)(b)(a)$. In this case the alternating components
   of $v$ are $u_1= aba$, $u_3= ab$, and $u_5= b$.
   
   \begin{prop}\label{prop:cell} Let $v\in \Aa ^+$ and $u_i$, $i= 1, 3, \ldots, k-1$,  be the alternating components of $v$.  Then one has:
\[\delta(v) = \ext(v)  - \sum_{i=0}^{\frac{k-2}{2}} \left \lceil \frac{|u_{2i+1}|}{2}\right\rceil.\]
   \end{prop}
   \begin{proof}Let $ v = v_0u_1v_2u_3 \cdots u_{k-1}v_k$. Since  $\delta(v_{2i}) = \ext(v_{2i})$, $0\leq i \leq k/2$ and $|u_{2i+1}|= \ext(u_{2i+1})$, $0\leq i\leq k/2-1$, one has 
\[ \delta(v) = \sum_{i=0}^{\frac{k}{2}} \delta(v_{2i}) + \sum_{i=0}^{\frac{k-2}{2}}\left \lfloor \frac{|u_{2i+1}|}{2}\right\rfloor =\]
\[ \sum_{i=0}^{\frac{k}{2}} \ext (v_{2i}) + \sum_{i=0}^{\frac{k-2}{2}} \left(\ext (u_{2i+1}) -\left \lceil \frac{|u_{2i+1}|}{2}\right\rceil \right) = \ext (v) - \sum_{i=0}^{\frac{k-2}{2}} \left \lceil \frac{|u_{2i+1}|}{2}\right\rceil. \qedhere\]
    \end{proof}
    
\smallskip
   
   \begin{example}
If $v= (a^3b^2)(aba)(b^2)(ab)(a^2)(b)(a)$, we have $\ext(v)= 11$, $\lceil |aba|/2\rceil = 2$, $\lceil |ab|/2\rceil = 1=\lceil |b|/2\rceil $, so that $\delta(v) = 11- 4 = 7$.
\end{example}
   
    In the following, for each word $v\in \Aa ^*$ we let $[v]$ denote the set $[v]= \{v, v\sptilde  , {\bar v}, {\bar v}\sptilde   \}$. From Proposition~\ref{prop:basicp} all Christoffel words
 $a\psi(z)b$ with a directive word $z\in [v]$ have the same length. The next proposition shows that they have the same depth.
 
 \begin{prop} \label{prop:Crisdepth}Let $v\in \Aa ^*$. All Christoffel words  $a\psi(z)b$ with $z\in [v]$ have the same depth.
 \end{prop}
 \begin{proof} The result is trivially true if $v$ is constant. Let us then suppose that $v$ is not  constant.  From  Propositions~\ref{prop:height}
 and~\ref{prop:height1} it is sufficient to prove that $\delta(z)= \delta(v)$ for all $z\in [v]$. It is readily verified that $\delta({\bar v}) = \delta(v)$
 as  $\ext({\bar v})= \ext(v)=n$ and for each   $0\leq i \leq n$, $\delta_i({\bar v}) = \delta_i(v)$. Let us now prove that $\delta(v\sptilde )= \delta(v)$. Let us write $v$ as in Eq.~\eqref{eq:sim0}, so that 
\[v\sptilde = v_k\sptilde u_{k-1}\sptilde  \cdots u_1\sptilde v_0\sptilde .\]
  Since   all $u_i\sptilde $, $i= k-1, k-3, \ldots, 3, 1$ are the \emph{alternating components}  of $v\sptilde $, by 
 the fact that $\ext(v\sptilde )= \ext(v)=n$ and $|u_i|= |u_i\sptilde |$ in view of   Proposition~\ref{prop:cell} it follows that  $\delta(v\sptilde )= \delta(v)$.  From the previous results it follows immediately that $\delta({\bar v}\sptilde ) = \delta(v)$.
  \end{proof}

Let $u=x_0^{\alpha_0}x_1^{\alpha_1}\cdots x_n^{\alpha_n}$. We define $u_1$ as:  $u_1=\varepsilon$  if  $\alpha_0>1$ or $n=0$, and, otherwise, $u_1$ is  the longest  proper  prefix of $u$ such that $u=u_1u'$, $u_1$ is alternating, and $u'^{(F)} \neq u_1^{(L)}$, i.e., the longest prefix  $x_0^{\alpha_0}x_1^{\alpha_1}\cdots x_i^{\alpha_i}$ of $u$ with $i<n$ and $\alpha_0=\alpha_1= \ldots = \alpha_i=1$.

 Similarly, we define $u_2$ as $u_2=\varepsilon$ if  $\alpha_n>1$ or $n=0$ and, otherwise, $u_2$ is  the longest  proper  suffix of $u$ such that $u=u'u_2$, $u_2$ is alternating, and $u'^{(L)} \neq u_2^{(F)}$. For instance, if $u= ab^2aba$, then $u_1=a$ and $u_2= aba$; if $u=a^2b$, then $u_1=\varepsilon$ and $v_2= b$.

 \begin{prop}\label{prop:u+v} Let  $u=x_0^{\alpha_0}x_1^{\alpha_1}\cdots x_n^{\alpha_n}$ and  $v=y_0^{\beta_0}y_1^{\beta_1}\cdots y_m^{\beta_m}$. Then
\[\delta (u)+\delta (v)-1\leq \delta (uv)\leq \delta (u)+\delta (v).\]
Moreover,   $\delta (uv) = \delta (u) +\delta (v)$ if and only if  $u^{(L)}\not= v^{(F)}$ and $|u_2|$, $|v_1|$ are both even.
\end{prop}
\begin {proof} If  $u^{(L)}= v^{(F)}$,  then $uv=x_0^{\alpha_0}x_1^{\alpha_1}\cdots x_{n-1}^{\alpha_{n-1}} x_n^{\alpha_n +\beta_0}y_1^{\beta_1}\dots y_m^{\beta_m}$ and trivially  $\delta (uv) = \delta (u) +\delta (v)-1$.  Let us then suppose  $u^{(L)}\not= v^{(F)}$.  We consider two cases: $|u_2|$ even and $|u_2|$ odd.\par
If  $|u_2|$ is even,  then  $\delta_i(u)=\delta_i(uv)$, $i=0, \dots, n$. If    $|v_1|=0$, then   $\delta_i(v)=\delta_{n+i+1}(uv)$, $i=0, \dots , m$ so that  $\delta (uv) = \delta (u) +\delta (v).$  Let then  $|v_1|=r\geq 1$.  For each  $i=0,\dots,r-1$, one has  $(\delta_i(v),\delta_{n+i+1}(uv))=(1,0)$ if $i$ is even and $(\delta_i(v),\delta_{n+i+1}(uv))=(0,1)$ if $i$ is odd. Moreover,   $\delta_i(v)=\delta_{n+i+1}(uv)$ for each  $i=r,\dots,m$.\par
It follows that if  $r$ is even, then  the number of pairs $(1,0)$ is equal to the number of pairs  $(0,1)$  that implies $\delta (uv) = \delta (u) +\delta (v)$.  If  $r$ is odd, then the number of pairs $(1,0)$  is equal to the number  of pairs $ (0,1)$ plus $1$, so that  $\delta (uv) = \delta (u) +\delta (v)-1$.\par
Let  $|u_2|$ be odd. In this case  $\delta_i(u)=\delta_i(uv)$ if  $i=0, \dots, n-1$, $\delta_n(u)=1$, and  $\delta_n(uv)=0$, $\delta_i(v)=\delta_{n+i+1}(uv)$, if  $i=0,\dots,m$. It follows  $\delta (uv) = \delta (u) +\delta (v)-1$ and then the assertion.  
\end {proof}

\begin{example}
Let $u= a^3b^2aba$, $w= a^3b^2a^2ba$, and $v= babab^2$. One has $u_2= aba$, $w_2= ba$, and $v_1= baba$.
One has $\delta(u)=\delta(w)=4$, and $\delta(v)= 3$. One verifies that $\delta(uv)=6$ and $\delta(wv)=7$.
\end{example}

From Proposition~\ref{prop:u+v} one readily derives:

\begin{cor}\label{cor: deltavx} Let  $u=x_0^{\alpha_0}x_1^{\alpha_1}\cdots x_n^{\alpha_n}$ and  $v= x\in \Aa $. Then
\[ \delta(u) \leq \delta(ux) \leq \delta(u)+1.\]
Moreover, $\delta(ux)= \delta(u)$ if and only if  $u^{(L)}=x$ or ${\bar u}^{(L)}=x$ and $|u_2|$  is odd.
\end{cor}

  \begin{lemma}\label{thm:h+}
If $v$ is a non-constant word, then  $h(v)=h(v_{+})+1$.
\end{lemma}
\begin{proof}
By the definition of height, we have $h(v)=h({}_{+}v)+1$. Moreover, Proposition~\ref{prop:Crisdepth} 
implies that $h(u)=h(u\sptilde)$ for any $u\in\Aa^{*}$. Hence, to obtain the assertion it suffices to observe by~\eqref{eq:E+} that $v_{+}=({}_{+}(v\sptilde))\sptilde$.
\end{proof}

We shall now give another equivalent definition for the function $h=\delta$.
Let $H:\Nn_{+}\to\Nn$ be the sequence defined by $H(1)=0$ and, for all $n>0$,
\[H(2n)=H(n)\quad\text{ and }\quad H(4n\pm 1)=H(n)+1\,.\]
The first few values of $H(n)$ are
\[0,0,1,0,1,1,1,0,1,1,2,1,2,1,1,0,1,\ldots.\]
As an immediate consequence of the definition, for any $k\geq 1$ we have
\begin{equation}
\label{eq:H}
H(2^{k+1}n\pm 1)=H(n)+1\,.
\end{equation}
Note that the sequence $\hat H$ given by $\hat H(n)=H(n)+1$ is the sequence A007302 in \cite{OEIS}.

\begin{prop}
For all $v\in\Aa^{*}$, $h(v)=H\left(\langle bvb\rangle\right)$.
\end{prop}
\begin{proof}
We proceed by induction on $h(v)$. If $h(v)=1$, then $v$ is constant, so that
$\langle bvb\rangle=2^{k+1}\pm 1$ for some $k\geq 1$. From~\eqref{eq:H},
$H\left(\langle bvb\rangle\right)=H(1)+1=1$ follows.

Let now $h(v)>1$, so that $v$ contains both $a$ and $b$ as letters. By Lemma~\ref{thm:h+} and 
the induction hypothesis, we have $h(v)=H\left(\langle b(v_{+})b\rangle\right)+1$.

Now, if $v^{(L)}=a$, then there exists $k\geq 1$ such that $v=(v_{+})ba^{k}$, so that
\[\langle bvb\rangle=
\langle b(v_{+})b\cdot a^{k}b\rangle=2^{k+1}\langle b(v_{+})b\rangle +1\,.\]
On the other hand, if $v^{(L)}=b$, then there exists $k\geq 1$ with $v=(v_{+})ab^{k}$, so that
\[\langle bvb\rangle=
\langle b(v_{+})ab^{k+1}\rangle=2^{k+1}\langle b(v_{+})b\rangle -1\,.\]
In both cases, by~\eqref{eq:H} it follows
$H\left(\langle bvb\rangle\right)=H\left(\langle b(v_{+})b\rangle\right)+1=h(v)$, as desired.
\end{proof}
\begin{example}
Let $v= a^2bab$. One has $\langle ba^2bab^2\rangle = 75$ and $H(75)= H(19)+1= H(5)+2= H(1)+3= 3$.
Hence, $h(v)= 3$.
\end{example}
\begin{prop}\label{prop:loup}
For all $v\in \Aa ^{*}$, \[\left\lfloor\frac{\ext (v)}{2}\right\rfloor+1\leq h(v)\leq\left\lfloor\frac{|v|}{2}\right\rfloor+1\,.\]
The set  of  words $v$  over $\Aa = \{a, b\}=\{x,y\}$ for which the lower bound  is attained is 
$Y= x^+(yx^+)^*$ if $\ext(v)$ is odd and $Y= \{\varepsilon\} \cup x^+ (yx^+)^*(y^+x)^*y^+$ if $\ext(v)$ is even.
 The set of  words for which  the upper bound  is attained is 
$X=\{ab,ba\}^{*}\{\varepsilon,a,b\}\{ab,ba\}^{*}$.
\end{prop}
\begin{proof}
Let us first prove the lower bound. Let $\ext(v)=n+1$. One has that $h(v)= \delta(v)= \ext(v)- \card\{0 <i< n \mid \delta_i(v)=0\}$. Since by the definition of $\delta$
in the sequence $\Delta_v= \delta_0(v)\cdots \delta_n(v)$ one cannot have two consecutive $0$, it follows that the maximal value of $ \card\{0 <i< n \mid \delta_i(v)=0\}$ is attained if and only if  $\Delta_v \in 1(01)^*$ if  $\ext(v)$ is odd and  $\Delta_v \in 1(01)^* (10)^*1$ if  $\ext(v)$ is even. In both the cases the previous maximal value is  equal to $\left \lceil \frac{n-1}{2}\right\rceil$. From this one has $\delta(v) \geq n+1 - \left \lceil \frac{n-1}{2}\right\rceil = \left\lfloor\frac{\ext (v)}{2}\right\rfloor+1$.  

To complete the first part of the proof it is sufficient to observe that   $\Delta_v \in 1(01)^*$
if and only if  $v\in x^+(yx^+)^*$,  and  $\Delta_v \in 1 (01)^*(10)^*1$ if and only if  $v\in x^+ (yx^+)^*(y^+x)^*y^+$.

Let us now prove the upper bound.
If $v$ is constant, then $h(v)=1$ and the result is trivially true. Let us then suppose that
$v$ is not constant. Let $n= h(v)>1$. By the definition of height, there exist
$v_{(1)}$,\dots$, v_{(n)}$  such that $v_{(1)}=v$, $v_{(n)}$ is constant,
and $v_{(i+1)}={}_{+}(v_{(i)})$ for $i=1,\ldots,n-1$. 
Since for any non-constant word $u$ one has $|u|\geq |_+u|+2$, it follows  $|v_{(i)}| \geq |v_{(i+1)}|+2$, $i=1, \ldots, n-1$, so that 
\begin{equation}\label{eq:lowupp}
 |v|\geq |v_{(n)}|+2n-2 \geq 2n-2; 
 \end{equation}
  hence
$n= h(v)\leq\lfloor|v|/2\rfloor+1$.

Let us now prove that $h(v)=\lfloor|v|/2\rfloor+1$ if and only if $v\in X$.
Suppose first that
$|v|$ is even. The set of words of even length within $X$ is $\{ab,ba\}^{*}$. Clearly, from (\ref{eq:lowupp}), 
$n=1+|v|/2$ if and only $v_{(n)}=\varepsilon$ and $|v_{(i)}|-|v_{(i+1)}|=2$ for $i=1,\ldots,n-1$. 
Now, $|u|-|{}_{+}u|=2$ if and only if $u=ab({}_{+}u)$ or $u=ba({}_{+}u)$.
It follows that $h(v)=1+|v|/2$ if and only
if $v\in\{ab,ba\}^{*}$.

Let now $|v|$ be odd. The subset of $X$ made by words of odd length is
\[X'=X\setminus\{ab,ba\}^{*}=\{ab,ba\}^*\{a,b\}\{ab,ba\}^*.\]

It is not difficult to see that
$X'= X'_1\cup X'_2$ where
\[X'_1=\{ab,ba\}^{*}\{a,b\} \ \  \mbox{and} \  \  X'_2 = \{ab,ba\}^{*}\{aab,bba\}\{ab,ba\}^{*}\,.\]
One has $|u|-|{}_{+}u|=3$ if and only if $u=aab({}_{+}u)$ or $u=bba({}_{+}u)$.
It follows that  if $v\in X'_1$, then $v_{(n)}\in\{a,b\}$ and  $|v_{(i)}|-|v_{(i+1)}|=2$ for $i=1,\ldots,n-1$.
If $v\in X'_2$, then $v_{(n)}=\varepsilon$ and 
$|v_{(i)}|-|v_{(i+1)}|=2$ for all $i$ in $\{1,\ldots,n-1\}$ except exactly one $j$ for which
$|v_{(j)}|-|v_{(j+1)}|=3$. From (\ref{eq:lowupp}) one has that in both cases
 $n=(|v|+1)/2=\lfloor |v|/2\rfloor+1$.

Conversely, from (\ref{eq:lowupp}) if $v$ is such that $n=(|v|+1)/2$, then we must have either $v_{(n)}\in\{a,b\}$
and $|v_{(i)}|-|v_{(i+1)}|=2$ for $i=1,\ldots,n-1$, or $v_{(n)}=\varepsilon$ and 
$|v_{(i)}|-|v_{(i+1)}|=2$ for all $i$ in $\{1,\ldots,n-1\}$ except exactly one $j$ for which
$|v_{(j)}|-|v_{(j+1)}|=3$. In the former case, we obtain $v\in X'_1$, and in
the latter $v\in X'_2$.
\end{proof}

We say that a word $v\in \Aa ^+$ is \emph{quasi-alternating} if each letter of $v$ but exactly one, is immediately followed by its complementary.  For instance, the words  $ab^2ab$ and $aba^2bab$ are quasi-alternating.

\begin{cor} Let $v\in \Aa ^+$. Then
\begin{equation}\label{eq:lowerupper}
 h(v)=\left\lfloor\frac{\ext (v)}{2}\right\rfloor+1=\left\lfloor\frac{|v|}{2}\right\rfloor+1\,
 \end{equation}
if and only if $v$ is alternating or $v$ is quasi-alternating with $\ext(v)$ equal to an even integer.
\end{cor}
\begin{proof} ($\Rightarrow$) If~\eqref{eq:lowerupper} is satisfied, then $|v|= \ext(v)$ or $|v|= \ext(v)+1$. In the first case $v$ is alternating and in the second case quasi-alternating. Moreover, in the latter case $\ext(v)$ has to be even, otherwise $\frac{|v|}{2} = \lfloor\frac{\ext (v)}{2}\rfloor+1$, a contradiction.

\noindent ($\Leftarrow$) If $v$ is alternating or quasi-alternating, then by the preceding proposition $v\in X$ so that 
 $h(v)= \lfloor\frac{|v|}{2}\rfloor+1$. Moreover, if $v$ is alternating, then $|v|= \ext(v)$ and we are done. If $v$ is quasi-alternating, then
 $|v|= \ext(v)+1$. If $\ext(v)$ is even,  then $\lfloor|v|/2\rfloor = \frac{\ext(v)}{2}$ and the result is obtained.
\end{proof}
\begin{example}
Let $v= a^2ba^3baba^3b$. One has $\Delta_v= 10101011$, so that $h(v)=\delta(v)=5$. Since $\ext(v)=8$, one has
that $v\in Y$ and $h(v)= \ext(v)/2 +1$. Let $v=abab^2a^2b \in X$; one has $h(v)=5= |v|/2+1$. Let $v$ be the quasi-alternating word
$v= abab^2ab$; one has $h(v)= \ext(v)/2+1=4= \lfloor |v|/2\rfloor+1.$
\end{example}

For each pair  $k, p$ of positive integers we let  $X_k(p)$ denote the set of all words of length $k$ having a height equal to $p$, i.e.,
\[X_k(p)= \{v\in \Aa ^k \mid h(v)= p\}.\] Moreover, we set  $J_k(p)= \card(X_k(p))$. From the definition one has $X_k(1)= \{a^k, b^k\}$. By Proposition~\ref{prop:loup} one has $X_k(p)= \emptyset$ if $p> \lfloor \frac{k}{2}\rfloor +1$.

In order to give an exact formula for $J_{k}(p)$, we need some notation and preparatory results.  We recall that for any $v =x_0^{\alpha_0}x_1^{\alpha_1}\cdots x_n^{\alpha_n}$, the word $v_2$ is defined as $v_2=\varepsilon$ if  $\alpha_n>1$ or $n=0$ and, otherwise, $v_2$ is  the longest  proper  suffix of $v$ such that $v=v'v_2$, $v_2$ is alternating, and $v'^{(L)} \neq v_2^{(F)}$. 

 Let $ E$ be the set of words $v$ such that $v_{2}$ is of even length, i.e.,
$ E=\{v\in\Aa^{*}\mid |v_{2}|\equiv 0 \pmod 2\}$, and let
\[e_{k}(p)=\card(X_{k}(p)\cap E),\quad o_{k}(p)=\card(X_{k}(p)\setminus E),\]
so that $J_{k}(p)=e_{k}(p)+o_{k}(p)$.

The following proposition gives a recursive procedure allowing to
computing $X_k(p)$ and then $J_k(p)$,  for all  $k$ and $p$. 

\begin{prop} For all $k>0$ and $p>0$,
\begin{align*}
X_{k+1}(p)\cap E&= \{vv^{(L)} \mid v\in X_k(p)\}\cup \{v\bar v^{(L)} \mid v\in X_k(p)\setminus E\},\\
X_{k+1}(p)\setminus E&= \{v\bar v^{(L)} \mid v\in X_k(p-1)\cap E\}.
\end{align*}
Hence,
\begin{align}
\label{eq:JE}
e_{k+1}(p)&= J_k(p) + o_{k}(p)=e_{k}(p)+2o_{k}(p), \\
\label{eq:JO}
o_{k+1}(p)&= e_{k}(p-1).
\end{align}
\end{prop}
\begin{proof} By Corollary~\ref{cor: deltavx}, for any $x\in \Aa $ one has  $ h(vx)= \delta(vx)= \delta(v)= h(v)=p$ if and only if  $v^{(L)}=x$ or ${\bar v}^{(L)}=x$ and $|v_2|$  is odd. If ${\bar v}^{(L)}=x$ and $|v_2|$  is even, then $h(v)= p-1$. Moreover, it is clear from the definition that $vv^{(L)}\in E$ for all $v$, whereas $v\bar v^{(L)}$ is in $ E$ if and only if $v$ is not. From this the result follows.
\end{proof}
\begin{example}
Since if $v\in X_k(p)$ then ${\bar v}\in X_k(p)$, we set $X'_k(p)= \{v \in X_k(p) \mid v^{(F)}=a\}$. For $k=2$ one has
$X'_2(1)=\{a^2\}$ and $X'_2(2)= \{ab\}$. For $k=3$, $X'_3(1)=\{a^3\}$ and $X'_3(2)= \{ab^2, a^2b, aba\}$. For $k=4$, $X'_4(1)= \{a^4\}$,
$X'_4(2)= \{ab^3, a^2b^2, aba^2\} \cup \{a^2ba\} \cup\{a^3b\}$, $X'_4(3)= \emptyset \cup \emptyset \cup\{ab^2a, abab\}$.
\end{example}

\begin{lemma}\label{lemma:ocappa}
For all $k>0$, one has $o_{k}(1)=0$ and for $p>1$, \[o_{k}(p)=2^{p-1}\binom{k-p}{p-2}\] 
with the usual convention that $\binom{n}{m}=0$ whenever $n<m$.
\end{lemma}
\begin{proof}
Clearly $o_{k}(1)=0$ since $X_{k}(1)=\{a^{k},b^{k}\}\subseteq E$. Moreover, from~\eqref{eq:JO} it follows $o_{k}(2)=e_{k-1}(1)=2=2^{2-1}\binom{k-2}{2-2}$.

Let now $p>2$. The assertion is trivially verified if $k<2(p-1)$, since this implies
$p>\lfloor k/2\rfloor+1$ and then $0=J_{k}(p)\geq o_{k}(p)$. If $k=2(p-1)$, we have $p=\lfloor k/2\rfloor +1$, so that $X_{k}(p)=\{ab,ba\}^{p-1}$ and $J_{k}(p)=2^{p-1}$.
By~\eqref{eq:JE}, $e_{k}(p)=J_{k-1}(p)+o_{k-1}(p)\leq 2J_{k-1}(p)=2J_{2p-3}(p)=0$; hence,
\[o_{k}(p)=J_{k}(p)=2^{p-1}=2^{p-1}\binom{k-p}{p-2}.\]

We can now assume, by (double) induction, that the assertion is verified for all smaller values of $k$ and $p$. Substituting~\eqref{eq:JE} in~\eqref{eq:JO}, we obtain
\[\begin{split}
o_{k}(p)=e_{k-1}(p-1)&=e_{k-2}(p-1)+2o_{k-2}(p-1)\\
&=e_{k-3}(p-1)+2o_{k-3}(p-1)+2o_{k-2}(p-1)\\
&=\cdots=2\sum_{i=2(p-2)}^{k-2}o_{i}(p-1),
\end{split}\]
where the last equality holds because $e_{2(p-2)}(p-1)=0$. Therefore, by induction we have
\[o_{k}(p)=2\sum_{i=2(p-2)}^{k-2}2^{p-2}\binom{i-p+1}{p-3}
=2^{p-1}\sum_{i=0}^{k-2(p-1)}\binom{i+p-3}{p-3}.\]
The assertion now follows from the identity (see, for instance, \cite{GKP})
\[\sum_{j=0}^{n-1}\binom{j+m}{m}=\binom{n+m}{m+1}.\qedhere\]
\end{proof}

\begin{thm}
For all $k,p>0$,
\[J_{k}(p)=
2^{p-1}\left (\binom{k-p+1}{p-1}+\binom{k-p}{p-1}\right ).\]
\end{thm}

\begin{proof}
If $k<2(p-1)$, then $p>\lfloor k/2\rfloor+1$, so that $J_{k}(p)=0$ as desired.

Let now $k\geq 2(p-1)$. The assertion is trivially verified for $p=1$, so let us suppose $p>1$. Using~\eqref{eq:JO} and Lemma \ref{lemma:ocappa}, we obtain
\[J_{k}(p)=e_{k}(p)+o_{k}(p)=o_{k+1}(p+1)+o_{k}(p)=2^{p}\binom{k-p}{p-1}+2^{p-1}\binom{k-p}{p-2}=2^{p-1}\left(2\binom{k-p}{p-1}+\binom{k-p}{p-2}\right).\]
The proof is completed by Pascal's rule.
\end{proof}
From the preceding theorem one derives a simple formula for the number of words of length $k$ for which the height reaches its maximal value $\lfloor\frac{k}{2}\rfloor +1$.
\begin{cor} Let $k>0$.  If $k$ is even, $J_k( \frac{k}{2} +1) = 2^{\frac{k}{2}}$ and if $k$ is odd, $J_k( \lfloor\frac{k}{2}\rfloor +1) = 2^{\frac{k+1}{2}}(1+\frac{1}{2}\lfloor \frac{k}{2}\rfloor)$.
\end{cor}

\section{Derivative of a standard word}\label{sec:ders}

In this section we shall see that any finite or infinite standard Sturmian word  $w$ has,  with respect to a suitable endomorphism of  $\Aa ^*$,  a derivative which is still a standard Sturmian word. 

For any $k\geq 0$
we define  
\[ X'_k=  \{a^kb, a^{k}ba\} \ \text{ and }\ Y'_k= \{b^ka, b^{k}ab\}.\]

The sets $X'_k$ and $Y'_k$ are  codes  having a finite deciphering delay \cite{codes}, so that any word $x\in X'^{\infty}_k$ (resp., $ x\in Y'^{\infty}_k$) can be uniquely factored by the elements of $X'_k$ (resp.,  $Y'_k$). 

Let  $w= uxy$, with $u\in \PER $ and $\{x,y\}= \{a, b\}$, be a proper standard Sturmian word. We define index of $w$ the index of the central word $u$.

\begin{lemma} Let $w$ be a  proper standard Sturmian word of index $k$.
Then  $w\in X'^*_k\cup \{a^{k+1}b\}$ if $w^{(F)}=a$ and  $w\in Y'^*_k \cup \{b^{k+1}a\}$ if $w^{(F)}=b$.
\end{lemma}
\begin{proof} Let $w=\psi(v)xy$ be a proper standard Sturmian word of index $k$. We suppose that $w^{(F)}=v^{(F)}=a$. The case $w^{(F)}= v^{(F)}= b$ is symmetrically dealt with. If $v$ is constant, i.e., $v=a^k$, then $vab = a^{k+1}b$ and $vba= a^kba \in X'_k$ and the result is achieved. Let us then assume that $v$ is not constant.  By Lemma~\ref{lem:AABC} one has that $a\psi(v)b \in a^{k+1}b\{a^kb, a^{k+1}b\}^*$, so that, as
$\psi(v)$ is a palindrome,
\[\psi(v) \in  a^{k}b\{a^kb, a^{k+1}b\}^*a^k.\]
As is readily verified $a^kb\{a^kb, a^{k+1}b\}^* \subseteq X'^*_k$, so that $\psi(v) \in X'^*_ka^k$. Hence, $\psi(v)ba \in X'^*_ka^kba \subseteq X'^*_k$. As $\psi(v)$ terminates with $a^kba^k$ it follows that $\psi(v)ab \in  X'^*_ka^kba^{k+1}b = X'^*_k(a^kba)a^{k}b\subseteq X'^*_k$, which concludes the proof.
\end{proof}

If $w$ is a proper standard Sturmian word, we can introduce a derivative of $w$ as follows. 
For each $k\geq 0$ if $w^{(F)}=a$ we consider the code $X'_k$ and
the injective endomorphism $\mu_k$= $\mu_{a^kb}$ of $\Aa ^*$ defined by 
\begin{equation}\label{eq:fi1}
 \mu_k(a) = a^{k}ba, \quad \mu_k(b) = a^{k}b.
\end{equation}
By the previous lemma if $w\in X'^*_k$, we define the derivative $Dw$ of $w$ equal to the derivative $D_kw$ with respect to $\mu_k$,
i.e.,   $D_k w = \mu_k^{-1}(w)$.  If $w= a^{k+1}b$, we
define $D a^{k+1}b  = D_{k+1} a^{k+1}b =b$. Let us observe that from the definition for all $k\geq 0$ one has $D_k a^{k}ba =a$.

If $w^{(F)}=b$, we consider the code $Y'_k$ and 
 the injective endomorphism $\hat\mu_k = \mu_{b^ka}$ of $\Aa ^*$ defined by 
\begin{equation}\label{eq:fi2}
\hat \mu_k(a) = b^ka, \quad \hat \mu_k(b) = b^{k}ab.
\end{equation}
By the previous lemma if $w\in Y'^*_k$ we define the derivative $Dw$ of $w$ equal to the derivative ${\hat D}_kw$ with respect to $\hat\mu_k$, i.e.,  $\hat{D}_k w = \hat\mu_k^{-1}(w)$.  If $w= b^{k+1}a$ we define $D b^{k+1}a = 
 \hat {D}_{k+1} b^{k+1}a =a$. Observe that  for all $k\geq 0$ one has $\hat{D}_k b^{k}ab =b$.

Finally, observe that if $k=0$, i.e., $w=ba$ or $w=ab$, from the previous definition one has  $D ba = a$ and $D ab = b$.

\begin{example}
\label{ex:index1}
Let $v=ab^2a^2b$ and $w$ be the standard  word $\psi(v)ba$ where $\psi(v)$ is a  central 
word of index 1. One has:
\[ w = ababaababaabababaababaabababa.\]
In this case one has $X'_1=\{aba, ab\}$ and
$D w = D_1w=  (bababbabab)ba = \psi(ba^2b)ba.$ Similarly, one has  $D \psi(v)ab= \psi(ba^2b)ab.$

If $w=\psi(b^2a^2)ba$, then $w=(bbabbabb)ba$. The index of $w$ is 2 and $Y'_2=\{b^2a, b^2ab\}$
and $D w = \hat{D}_2w=  aba$. Similarly, one has $D\psi(b^2a^2)ab = aab$.
\end{example}
Let us recall (cf. \cite{LO2, adl97}) that an endomorphism  $f$ of  $\Aa ^*$ is called a \emph{standard Sturmian morphism} if the image $f(s)$ of  any finite or infinite standard Sturmian word $s$ is a standard Sturmian word.  This implies that  if  the image  $f(s)$ of a binary word $s\in \Aa ^{\infty}$  is a standard Sturmian word so is $s$. As  is well-known standard Sturmian morphisms form a monoid generated by the morphisms $\mu_a$, $\mu_b$, and $E$. Hence, 
for each $k\geq 0$,  $\mu_k, \hat\mu_k \in \{\mu_a, \mu_b\}^*$ are standard Sturmian morphisms called {\em pure}.

\begin{thm}\label{thm:diretto1} Let $k>0$ and $w\in X'^*_k\cup Y'^*_k \cup \{a^{k+1}b\} \cup \{b^{k+1}a\}$. Then $w$ is a proper standard Sturmian word if and only if $D w$ is a standard Sturmian word.
\end{thm}
\begin{proof} ($\Rightarrow$) We shall suppose without loss of generality that $w^{(F)}=a$, so that, as $w$ is a proper standard Sturmian word, $w\in X'^*_k\cup \{a^{k+1}b\}$. If $w\in X'^*_k$, then $Dw= \mu_k^{-1}(w)$. Since $\mu_k$ is a standard Sturmian morphism,  it follows
that $Dw \in \Stand $ . Similarly, if $w= a^{k+1}b$, then $Dw = \mu_{k+1}^{-1}(a^{k+1}b)=b \in \Stand $.

($\Leftarrow$) Let us now suppose that $Dw \in \Stand $. If $w\in X'^*_k$, then, as $\mu_k$ is a standard Sturmian morphism, $\mu_k(Dw)=w$
is a standard and proper Sturmian word. Similarly, if $w= a^{k+1}b$ one has $\mu_{k+1}(Dw)= \mu_{k+1}(b)= w$ which is a proper standard Sturmian word.
\end{proof}

The following theorem relates, through their directive words, the central word of a proper standard word and the central word of its derivative.

\begin{thm}\label{cor:vpius} Let $w= \psi(v)xy$, with $\{x,y\}= \{a,b\}$ and $v$ a non-constant word. Then
\[ D\psi(v)xy = \psi(_+v)xy.\]
\end{thm}

\begin {proof} Let $v= x_0^{\alpha_0}\cdots  x_n^{\alpha_n}$ with $n\geq 1$ as $v$ is not constant.  We can write, setting $\alpha_0=k$ and $\alpha_1=h$,  $v=x_0^kx_1^h\xi$.  We shall first suppose that $\psi(v)\in \PER_a$, so that 
  $v=a^kb^h\xi$. Since $_+v= b^{h-1}\xi$, by using  the Justin formula   we can write:
  \[\psi(v)= \psi(a^kb^h\xi)= \mu_{a^kb}(\psi(b^{h-1}\xi))\psi(a^kb)= \mu_{a^kb}(\psi(_+v))a^kba^k,\]
so that if $x=b$ and $y=a$
\[ w_1 = \psi(v)ba =  \mu_{a^kb}(\psi(_+v))(a^kb)(a^kba),\]
and if $x=a$ and $y=b$
\[ w_2 = \psi(v)ab =  \mu_{a^kb}(\psi((_+v))(a^kba)(a^kb).\]
Hence,
\[D w_1 = \psi(_+v)ba,\]
and
\[Dw_2= \psi(_+v))ab.\] which concludes the proof in the case $\psi(v)\in \PER_a$. The case $\psi(v)\in \PER_b$ can be proved in a similar way.  \end {proof}

\begin{example}
\label{ex:index1}
Let $w=uba$ with $u=\psi(a^2b^2a)$. One has
\[w= aabaabaaabaabaaba\] and
$D w = babba= \psi(ba)ba$. If $w= \psi(ba^2b^2a)ab$, one easily obtains $D w = \psi(ab^2a)ab$. If $w=\psi(abab)ba$, one derives $D w = \psi(ab)ba$.
\end{example}

\begin{cor} Let $w$ be the  standard word $w= \psi(v)ba$ where $v$ is not  constant  and $w'$ is  the Christoffel word $w'= a\psi(v)b$. Then
\[ Dw= D \psi(v)ba  = a^{-1}\partial(a\psi(v)b) a = a^{-1}(\partial w') a.\]
\end{cor}
\begin{proof} By the preceding theorem  $D \psi(v)ba  = \psi(_+v)ba$. By Theorem~\ref{thm: derpart} one has $\partial a\psi(v)b = a\psi(_+v)b$.
From this the result follows.
\end{proof}

\begin{remark}
The preceding corollary holds true also for the constant words $a^k$, $k\geq 0$. Indeed, $D a^kba = a = \partial a^{k+1}b $.
However, it is not more true for $b^k$, $k>0$. Indeed, $D b^{k+1}a = a$ whereas $\partial ab^{k+1} = b$.
\end{remark}

 From Theorem~\ref{thm:diretto1} any  proper standard Sturmian word $w$ has a derivative $w'=  Dw \in \Stand $. Therefore, if $w'$ is proper one can consider $D w' \in \Stand $  that we shall denote $D^2 w$.
In general, for any $p\geq 1$, $D^p w$ will  denote the derivative of order $p$ of $w$. Since  $|D^p w| > |D^{p+1} w|$, there exists an integer $d$ such that $D^d w \in \Aa $;
we call $d$ the \emph{depth} of the standard word $w$.

\begin{example}
 Let $w$ be the standard word $w=\psi(a^2b^2a)ba$ of  Example~\ref{ex:index1}. One has
$Dw = babba= \psi(ba)ba$, $D^2w = ba$, and $D^3 w= a$. Thus  the depth of $w$ is $d=3$.
\end{example}

\begin{prop} The depth of a standard word $w=\psi(v)xy$ with $\{x, y\}= \{a, b\}$ is equal to the depth of the Christoffel word $a\psi(v)b$.
\end{prop}
\begin{proof} Let $n=h(v)$ be the height of $v$. The result is trivially true if $v$ is constant or, equivalently, if  $n=1$. Let us then suppose that $v$ is not constant.  From Corollary~\ref{cor:vpius} one derives that for all $p\leq n$
\[D^{p-1} \psi(v)xy  = \psi(v_{(p)})xy.\]
Hence, $D^{n-1} \psi(v)xy  = \psi(v_{(n)})xy.$ Since $v_{(n)}$ is constant, one has $\psi(v_{(n)}) = v_{(n)}$ and $D v_{(n)}xy \in \Aa $.
Thus, the depth of $w$ is equal to $n=h(v)$ and by Proposition~\ref{prop:height}  is equal to the depth of $a\psi(v)b$.
\end{proof}

Let $s$ be now a characteristic, or infinite standard, Sturmian word. As we have seen in Sect.~\ref{sec:due},   \[s= \psi(v) \text{ with } v\in \Aa ^{\omega}\setminus \Aa ^*(a^{\omega} \cup b^{\omega}).\]
Any word $ v\in \Aa ^{\omega}\setminus \Aa ^*(a^{\omega} \cup b^{\omega})$ can be uniquely represented as:
\begin{equation}\label{eq:repr1}
v=x_0^{\alpha_0}x_1^{\alpha_1}x_2^{\alpha_2}\cdots x_{n-1}^{\alpha_{n-1}}x_n^{\alpha_n} \cdots
\end{equation}
   where for $i\geq 0$, $\alpha_i \geq 1$, $x_i\in \Aa $, 
  and  $x_{i+1}= {\bar x_i}$.

We define \emph{index} of the characteristic word $s= \psi(v)$ the first  exponent in the  representation~\eqref{eq:repr1} of $v$, i.e., $\alpha_0$.  We let  $\ind(s)$ denote the index of $s$. 
\begin{lemma}\label{lem:st000} Let $s$ be a characteristic Sturmian word of index $k$. Then $s\in X'^{\omega}_k$ if  $s^{(F)}=a$ and $s\in Y'^{\omega}_k$ if $s^{(F)}=b$.
\end{lemma}
\begin{proof} Let $s=\psi(v)$.  We first suppose that $s^{(F)}=a$. Since $s$ has index $k$, we can write $v= a^kbv'$ with $v'\in \Aa ^{\omega}$.
By Lemma~\ref{Ju00} one has:
\[s= \psi(a^kbv') = \mu_{a^kb}(\psi(v')).\]
From~\eqref{eq:fi1}, it follows that $s\in X'^{\omega}_k$. In a similar way one proves that $s\in Y'^{\omega}_k$ if $s^{(F)}=b$.
\end{proof}

We can now define the derivative $Ds$ of a characteristic Sturmian word $s$ of index $k$ as follows: 
\[ Ds = \mu_k^{-1}(s) \text{ if }s^{(F)}=a, \  Ds = {\hat\mu}_k^{-1}(s) \text{ if }s^{(F)}=b .\]
\begin{remark}
As one easily verifies,  $Ds$  is word isomorphic to the derived word of $s$ in the sense of  Durand \cite{Du} constructed by factoring $s$ in terms of the \emph{first returns} to the prefix of length $k+1$ of $s$. If $s^{(F)}=a$ (resp., $ s^{(F)}=b$) the set of first returns to the prefix $a^kb$
(resp., $b^ka$) of $s$ is  $\{a^kb, a^kba\}$ (resp., $\{b^ka, b^kab\}$). 
We mention that a further notion of derivative for  infinite words admitting a prefixal factorization, such as the characteristic Sturmian words, has been recently given in \cite{DZ15}.
\end{remark}

\begin{thm}\label{thm:000001} Let $s= D t$ with $t\in X'^{\omega}_k\cup Y'^{\omega}_k$.  Then $s$ is a characteristic Sturmian word if and only if so is $t$.
\end{thm}
\begin{proof} The result is an immediate consequence of the fact that the morphisms $\mu_k$ and ${\hat\mu}_k$ are standard Sturmian morphisms.
\end{proof}

Recall  that an  infinite word $v\in \Aa ^{\omega}$ is \emph{constant} if $v= x^{\omega}$ with $x\in \Aa $. If $v$ is not constant one can consider the greatest suffix $_+v$ of $v$ with respect to the suffixal ordering,  which is immediately preceded by a letter different from $v^{(F)}$.

\begin{thm}\label{thm:der0001}Let $w= \psi(v)$ be a characteristic  Sturmian word. Then
\[D\psi(v) = \psi(_+v).\]
\end{thm}  
\begin{proof} Let us write $v$ as  $v= x_0^{\alpha_0}x_1^{\alpha_1}\xi$ and  set  $k=\alpha_0$ and $h=\alpha_1$. We  first suppose that  $w^{(F)}=a$.  We can write $v=a^kb^h\xi$. Hence,  $\psi(v)= \psi(a^kb^h\xi)$. By Lemma~\ref{Ju00} we can write:
\[\psi(v)= \psi(a^kbb^{h-1}\xi)= \mu_{a^kb}(\psi(b^{h-1}\xi)) = \mu_k(\psi(b^{h-1}\xi)).\]
Hence, $D\psi(v) = \psi(b^{h-1}\xi)= \psi(_+v). $ If $w^{(F)}=b$ one has $v= b^ka^h\xi.$ We can write:
\[ \psi(v)= \psi(b^ka^h\xi)= \psi(b^kaa^{h-1}\xi)= \mu_{b^ka}(\psi(a^{h-1}\xi)) = {\hat \mu}_k(\psi(a^{h-1}\xi)).\]
Hence, $D\psi(v)= \psi(a^{h-1}\xi)= \psi(_+v)$.
\end{proof}

Let $s=\psi(v)$ be a characteristic Sturmian word. From Theorems \ref{thm:000001} and \ref{thm:der0001}, $Ds= \psi(_+v)$ is a characteristic  Sturmian word, so that $_+v$ is not  constant.  We can consider  the infinite sequence $(D^p s)_{p\geq 0}$ of successive derivatives of $s$
where
\[D^0s = s \ \text{ and } D^ps= D(D^{p-1}s), \text{ for }p>0.\]
 Similarly to the finite case,  one can introduce a sequence of infinite words $(v_{(n)})_{n>0}$, where $v_{(1)}=v$,
and for all $n\geq 1$, $v_{(n+1)} =  {_+(v_{(n)})}$.
If $s=\psi(v)$, then by the preceding theorem one has for each $p\geq 0$, $D^ps= \psi(v_{(p+1)})$ having for all $k>0$, $v= u_kv_{(k)}$ with $u_k\in\Aa ^*$ and $|u_k|<|u_{k+1}|$.

We say that a characteristic Sturmian word $s$ is \emph{stable} if there exist $m,n\geq 0$, $m\neq n$, such that $D^ms= D^ns$, i.e., $\card\{D^ms \mid m\geq 0\}<\infty$.

\begin{thm}\label{SUP} A characteristic Sturmian word $s$ is stable if and only if its directive word is ultimately periodic.
\end{thm}
\begin{proof} $(\Rightarrow)$ Let $m$ be the first integer $> 0$ such that there exists $n>m$ for which $D^{m-1}s= D^{n-1}s$. Hence,
$\psi(v_{(m)})= \psi(v_{(n)})$. Since $\psi$ is injective, it follows $v_{(m)}=v_{(n)}$ and therefore $v= u_mv_{(m)}= u_nv_{(n)}= u_nv_{(m)}$.  As $|u_m|<|u_n|$
one has $u_n= u_m\zeta$ with $\zeta\in \Aa ^+$ and $v_{(m)}=\zeta v_{(m)}$, so that $v_{(m)}$ is the periodic word $v_{(m)}= \zeta^{\omega}$ and
$v= u_m\zeta^{\omega}$.

$(\Leftarrow)$ Suppose that $s= \psi(v)$ with $v= pq^{\omega}$,  $p,q\in \Aa ^*$, and $q\neq \varepsilon$. There exists an integer
$k$ such that for all $j>k$, $D^{j-1}s= \psi(v_{(j)})$ with $v_{(j)}$ suffix of $q^{\omega}$. Hence, $v_{(j)}= q_j^{\omega}$, where $q_j$ is a conjugate of $q$.  By the pigeonhole principle  it follows that there exist two distinct integers $m,n >k$ such that $q_n=q_m$ and therefore $v_{(m)}=v_{(n)}$. This implies
$D^{n-1}s= D^{m-1}s$.
\end{proof}
\begin{example}
Let $f= \psi((ab)^{\omega}$) be the Fibonacci word. One has that for all $p\geq 1$, $D^pf = f$, so that $f$ is stable.
Let $s= \psi(a^k(ab)^{\omega})$ where $k$ is a fixed integer  $\geq 1$. One has $Ds= \psi((ab)^{\omega})= f$. Thus $Ds= D^ps =f$ for all $p\geq 1$ and $s$ is stable. Let $s= \psi(aba^2ba^3b\cdots ba^nb \cdots)$. For any $p>0$ one has $D^ps= \psi(a^{p+1}ba^{p+2}ba^{p+3}\cdots)$, so that $s$ is not stable.
\end{example}
Let $v=x_0^{\alpha_0}x_1^{\alpha_1}x_2^{\alpha_2}\cdots x_{n-1}^{\alpha_{n-1}}x_n^{\alpha_n} \cdots$ and $s= \psi(v)$ be  the characteristic Sturmian word with the directive word $v$. The slope of $s$ is  the limit  $\lim_{n\rightarrow \infty} \frac{|s_{[n]}|_b}{|s_{[n]}|_a}$. As is well-known (cf. \cite{LO2, BDL}),  since $s$ is Sturmian, this limit exists and is an irrational number  equal to the continued fraction
\[[\alpha_0; \alpha_1, \ldots, \alpha_n, \ldots ].\]
One can easily prove that the directive word $v$ is periodic if and only if there exist integers $r>0$ and $q\geq 0$ such that $\alpha_n= \alpha_{n+r}$ for all $n\geq q$, or,  equivalently,  the previous continued fraction is periodic. From Theorem~\ref{SUP} and \cite[Theorem 20]{AB}, one derives that \emph{a characteristic Sturmian word $s$ is stable if and only if the set of all derivated words in the sense of Durand (with respect to prefixes of $s$) is finite.}

\medskip

For each $k\geq 0$, let $X_k$ and $Y_k$ be the sets defined by~\eqref{eq:xy}.
\begin{lemma} Let $s$ be a characteristic Sturmian word of index $k$. Then $s\in X_k^{\omega}$ if $s^{(F)}=a$ and $s\in Y_k^{\omega}$ if $s^{(F)}=b$.
\end{lemma}
\begin{proof} Let us suppose $s^{(F)}=a$. As one readily verifies, for each $k\geq 0$ one has $X'^{\omega}_k =  \{a^kba, a^{k}b\}^{\omega}=
a^kb\{a^kb, a^{k+1}b\}^{\omega}= a^kb X^{\omega}_k$. By Lemma~\ref{lem:st000}, one has $s\in X'^{\omega}_k $, so that $s\in X_k^{\omega}$. The case $s^{(F)}=b$ is dealt with in a similar way.
\end{proof}

We can define the derivative $\partial s$ of a characteristic Sturmian word $s$ of index $k$  by
\[ \partial s = \varphi_k^{-1} (s) \text{ if }\ s^{(F)}= a \ \text{ and }\partial s = {\hat\varphi}_k^{-1} (s)\text{ if }\ s^{(F)}= b. \]

\begin{lemma} Let $x= x_1x_2\cdots x_n\cdots $ be an infinite word over $\Aa $. Then for each $k\geq 0$
\[ \varphi_k^{-1} (\mu_k(x))= bx.\]
\end{lemma}
\begin{proof} By Lemma~\ref{lemma:prepl}, one has for all $n\geq 1$
\[\varphi_k^{-1} (\mu_k(x_{[n]}b)) =\varphi_k^{-1} (\mu_k(x_{[n]})a^kb)=\varphi_k^{-1} (\mu_k(x_{[n]})b = bx_{[n]}.\]
 Thus,
\[\varphi_k^{-1} (\mu_k(x_{[n]}) = bx_{[n]}b^{-1}\]
 and
\[ \varphi_k^{-1} (\mu_k(x))= \lim_{n\rightarrow \infty}\varphi_k^{-1} (\mu_k(x_{[n]})=  \lim_{n\rightarrow \infty}bx_{[n]}b^{-1}= bx. \qedhere\]
\end{proof}

\begin{thm} Let $s$ be a characteristic  Sturmian word. Then
\[ \partial s = bDs. \]
\end{thm}
\begin{proof} If $s$ is a characteristic Sturmian word of index $k$, then by the preceding lemma one has:
\[\partial(\mu_k(Ds))= \partial s = bDs.\qedhere\]
\end{proof}

\section{Concluding remarks}

We have studied new combinatorial properties of 
Christoffel, central, and standard words,
which are related to a suitable notion of derivative of a word. In this analysis, the palindromization map that allows to construct all central words, as well as all infinite standard words, plays an essential role. Indeed,
it allows one to give a unified treatment for the previous classes of words. Moreover, one can make use of the important combinatorial tool represented by Justin's formula which links the palindromization map with pure standard morphisms.  By this palindromization map, from one side one can obtain a very simple formula giving the derivative of a Christoffel word. From the other one can extend the previous results to the case of standard words. Finally, new interesting combinatorial problems arose from considering higher order derivatives and the depth of a Christoffel word and of a standard word. This gives a new insight on these noteworthy classes of words.

An interesting open problem is to try to extend some of the previous results  to the case of alphabets with more than two letters, i.e., to the case of standard episturmian words. This extension seems to be quite hard since some basic combinatorial properties hold only for a binary alphabet.

\section*{Acknowledgments}
We would like to thank the anonymous referees for many valuable comments and suggestions.

\bibliographystyle{model1-num-names}

\begin{thebibliography}{99}
\expandafter\ifx\csname natexlab\endcsname\relax\def\natexlab#1{#1}\fi
\providecommand{\bibinfo}[2]{#2}
\ifx\xfnm\relax \def\xfnm[#1]{\unskip,\space#1}\fi
\bibitem[{Allouche and Shallit(2003)}]{AS}
\bibinfo{author}{J.-P. Allouche}, \bibinfo{author}{J.~Shallit},
  \bibinfo{title}{Automatic Sequences}, \bibinfo{publisher}{Cambridge
  University Press}, \bibinfo{address}{Cambridge UK}, \bibinfo{year}{2003}.
\bibitem[{Ara{\'u}jo and Bruyere(2005)}]{AB}
\bibinfo{author}{I.~Ara{\'u}jo}, \bibinfo{author}{V.~Bruyere},
\newblock \bibinfo{title}{Words derivated from {Sturmian} words},
\newblock \bibinfo{journal}{Theoret.\ Comput.\ Sci.} \bibinfo{volume}{340}
  (\bibinfo{year}{2005}) \bibinfo{pages}{204--219}.
\bibitem[{Berstel and de~Luca(1997)}]{BDL}
\bibinfo{author}{J.~Berstel}, \bibinfo{author}{A.~de~Luca},
\newblock \bibinfo{title}{Sturmian words, {Lyndon} words and trees},
\newblock \bibinfo{journal}{Theoret.\ Comput.\ Sci.} \bibinfo{volume}{178}
  (\bibinfo{year}{1997}) \bibinfo{pages}{171--203}.
\bibitem[{Berstel et~al.(2008)Berstel, Lauve, Reutenauer, and Saliola}]{BLRS}
\bibinfo{author}{J.~Berstel}, \bibinfo{author}{A.~Lauve},
  \bibinfo{author}{C.~Reutenauer}, \bibinfo{author}{F.~Saliola},
  \bibinfo{title}{Combinatorics on Words: Christoffel Words and Repetition in
  Words}, volume~\bibinfo{volume}{27} of \textit{\bibinfo{series}{CRM monograph
  series}}, \bibinfo{publisher}{American Mathematical Society},
  \bibinfo{address}{Providence, RI}, \bibinfo{year}{2008}.
\bibitem[{Berstel and Perrin(1985)}]{codes}
\bibinfo{author}{J.~Berstel}, \bibinfo{author}{D.~Perrin},
  \bibinfo{title}{Theory of Codes}, \bibinfo{publisher}{Academic Press},
  \bibinfo{address}{New York}, \bibinfo{year}{1985}.
\bibitem[{Berth\'e et~al.(2008)Berth\'e, de~Luca, and Reutenauer}]{BDR}
\bibinfo{author}{V.~Berth\'e}, \bibinfo{author}{A.~de~Luca},
  \bibinfo{author}{C.~Reutenauer},
\newblock \bibinfo{title}{On an involution of {C}hristoffel words and
  {S}turmian morphisms},
\newblock \bibinfo{journal}{European J. Combin.} \bibinfo{volume}{29}
  (\bibinfo{year}{2008}) \bibinfo{pages}{535--553}.
\bibitem[{Borel and Laubie(1993)}]{BL}
\bibinfo{author}{J.-P. Borel}, \bibinfo{author}{F.~Laubie},
\newblock \bibinfo{title}{Quelques mots sur la droite projective r{\'e}elle},
\newblock \bibinfo{journal}{J. Th{\'e}or. Nombres Bordeaux} \bibinfo{volume}{5}
  (\bibinfo{year}{1993}) \bibinfo{pages}{23--51}.
\bibitem[{Bucci et~al.(2009)Bucci, de~Luca, and De~Luca}]{BdD}
\bibinfo{author}{M.~Bucci}, \bibinfo{author}{A.~de~Luca},
  \bibinfo{author}{A.~De~Luca},
\newblock \bibinfo{title}{Characteristic morphisms of generalized episturmian
  words},
\newblock \bibinfo{journal}{Theoret.\ Comput.\ Sci.} \bibinfo{volume}{410}
  (\bibinfo{year}{2009}) \bibinfo{pages}{2840--2859}.
\bibitem[{Carpi and de~Luca(2004)}]{CdL}
\bibinfo{author}{A.~Carpi}, \bibinfo{author}{A.~de~Luca},
\newblock \bibinfo{title}{Codes of central {Sturmian} words},
\newblock \bibinfo{journal}{Theoret.\ Comput.\ Sci.}
  \bibinfo{volume}{340} (\bibinfo{year}{2005}) \bibinfo{pages}{220--239}.
\bibitem[{Christoffel(1875)}]{CFF}
\bibinfo{author}{E.~B. Christoffel},
\newblock \bibinfo{title}{Observatio arithmetica},
\newblock \bibinfo{journal}{Ann.\ Mat.\ Pur.\ Appl.} \bibinfo{volume}{6}
  (\bibinfo{year}{1875}) \bibinfo{pages}{148--152}.
\bibitem[{de~Bruijn(1981)}]{dB-Seq}
\bibinfo{author}{N.~G. de~Bruijn},
\newblock \bibinfo{title}{Sequences of zeros and ones generated by special
  production rules},
\newblock \bibinfo{journal}{Nederl. Akad. Wetensch. Indag. Math.}
  \bibinfo{volume}{43} (\bibinfo{year}{1981}) \bibinfo{pages}{27--37}.
\bibitem[{de~Bruijn(1989)}]{dB-gen}
\bibinfo{author}{N.~G. de~Bruijn},
\newblock \bibinfo{title}{Updown generation of {B}eatty sequences},
\newblock \bibinfo{journal}{Nederl. Akad. Wetensch. Indag. Math.}
  \bibinfo{volume}{51} (\bibinfo{year}{1989}) \bibinfo{pages}{385--407}.
\bibitem[{de~Luca(1997)}]{adl97}
\bibinfo{author}{A.~de~Luca},
\newblock \bibinfo{title}{Standard {Sturmian} morphisms},
\newblock \bibinfo{journal}{Theoret.\ Comput.\ Sci.} \bibinfo{volume}{178}
  (\bibinfo{year}{1997}) \bibinfo{pages}{205--224}.
\bibitem[{de~Luca(1997)}]{deluca}
\bibinfo{author}{A.~de~Luca},
\newblock \bibinfo{title}{Sturmian words: structure, combinatorics, and their
  arithmetics},
\newblock \bibinfo{journal}{Theoret.\ Comput.\ Sci.} \bibinfo{volume}{183}
  (\bibinfo{year}{1997}) \bibinfo{pages}{45--82}.
\bibitem[{de~Luca(2012)}]{SC}
\bibinfo{author}{A.~de~Luca},
\newblock \bibinfo{title}{A standard correspondence on epicentral words},
\newblock \bibinfo{journal}{European J. Combin.} \bibinfo{volume}{33}
  (\bibinfo{year}{2012}) \bibinfo{pages}{1514--1536}.
\bibitem[{de~Luca and De~Luca(2015)}]{adlADL}
\bibinfo{author}{A.~de~Luca}, \bibinfo{author}{A.~De~Luca},
\newblock \bibinfo{title}{Sturmian words and the Stern sequence},
\newblock \bibinfo{journal}{Theoret.\ Comput.\ Sci.} \bibinfo{volume}{581}
  (\bibinfo{year}{2015}) \bibinfo{pages}{26--44}.
\bibitem[{de~Luca and Mignosi(1994)}]{DM}
\bibinfo{author}{A.~de~Luca}, \bibinfo{author}{F.~Mignosi},
\newblock \bibinfo{title}{Some combinatorial properties of {Sturmian} words},
\newblock \bibinfo{journal}{Theoret.\ Comput.\ Sci.} \bibinfo{volume}{136}
  (\bibinfo{year}{1994}) \bibinfo{pages}{361--385}.
\bibitem[{de~Luca and Zamboni(2010{\natexlab{a}})}]{adlZ1}
\bibinfo{author}{A.~de~Luca}, \bibinfo{author}{L.~Q. Zamboni},
\newblock \bibinfo{title}{On graphs of central episturmian words},
\newblock \bibinfo{journal}{Theoret.\ Comput.\ Sci.} \bibinfo{volume}{411}
  (\bibinfo{year}{2010}{\natexlab{a}}) \bibinfo{pages}{70--90}.
\bibitem[{de~Luca and Zamboni(2010{\natexlab{b}})}]{adlZ2}
\bibinfo{author}{A.~de~Luca}, \bibinfo{author}{L.~Q. Zamboni},
\newblock \bibinfo{title}{Involutions of epicentral words},
\newblock \bibinfo{journal}{European J. Combin.} \bibinfo{volume}{31}
  (\bibinfo{year}{2010}{\natexlab{b}}) \bibinfo{pages}{867--886}.
\bibitem[{de~Luca and Zamboni(2015)}]{DZ15}
\bibinfo{author}{A.~de~Luca}, \bibinfo{author}{L.~Q. Zamboni},
\newblock \bibinfo{title}{On prefixal factorizations of words},
\newblock \bibinfo{journal}{European J. Combin.}, \bibinfo{volume}{52} (\bibinfo{year}{2016}) \bibinfo{pages}{59--73}.  
\bibitem[{Droubay et~al.(2001)Droubay, Justin, and Pirillo}]{DJP}
\bibinfo{author}{X.~Droubay}, \bibinfo{author}{J.~Justin},
  \bibinfo{author}{G.~Pirillo},
\newblock \bibinfo{title}{Episturmian words and some constructions of de {Luca}
  and {Rauzy}},
\newblock \bibinfo{journal}{Theoret.\ Comput.\ Sci.} \bibinfo{volume}{255}
  (\bibinfo{year}{2001}) \bibinfo{pages}{539--553}.
\bibitem[{Durand(1998)}]{Du}
\bibinfo{author}{F.~Durand},
\newblock \bibinfo{title}{A characterization of substitutive sequences using
  return words},
\newblock \bibinfo{journal}{Discrete Math.} \bibinfo{volume}{179}
  (\bibinfo{year}{1998}) \bibinfo{pages}{89--101}.
\bibitem[{Graham et~al.(1994)}]{GKP}
\bibinfo{author}{R.L.~Graham}, \bibinfo{author}{D.E.~Knuth},  \bibinfo{author}{O.~Patashnik}, \bibinfo{title}{Concrete Mathematics}, \bibinfo{edition}{2nd ed.}, \bibinfo{publisher}{Addison-Wesley}, \bibinfo{address}{Reading Mass.}, \bibinfo{year}{1994}.
\bibitem[{Justin(2005)}]{J}
\bibinfo{author}{J.~Justin},
\newblock \bibinfo{title}{Episturmian morphisms and a {Galois} theorem on
  continued fractions},
\newblock \bibinfo{journal}{Theor. Inform. Appl.} \bibinfo{volume}{39}
  (\bibinfo{year}{2005}) \bibinfo{pages}{207--215}.
\bibitem[{Klette and Rosenfeld(2004)}]{DSS} \bibinfo{author}{R.~Klette}, \bibinfo{author}{A.~Rosenfeld},
\newblock \bibinfo{title}{Digital straightness --- a review},
\newblock \bibinfo{journal}{Discrete Appl.\ Math.} \bibinfo{volume}{139}
  (\bibinfo{year}{2004}) \bibinfo{pages}{197--230}.
\bibitem[{Lothaire(1983)}]{LO}
\bibinfo{author}{M.~Lothaire}, \bibinfo{title}{Combinatorics on Words},
  \bibinfo{publisher}{Addison-Wesley}, \bibinfo{address}{Reading MA},
  \bibinfo{year}{1983}. \bibinfo{note}{Reprinted by Cambridge University Press,
  Cambridge UK, 1997}.
\bibitem[{Lothaire(2002)}]{LO2}
\bibinfo{author}{M.~Lothaire}, \bibinfo{title}{Algebraic Combinatorics on
  Words}, volume~\bibinfo{volume}{90} of \textit{\bibinfo{series}{Encyclopedia
  of Mathematics and its Applications}}, \bibinfo{publisher}{Cambridge
  University Press}, \bibinfo{address}{Cambridge UK}, \bibinfo{year}{2002}.
\bibitem[{Morse and Hedlund(1940)}]{MH}
\bibinfo{author}{M.~Morse}, \bibinfo{author}{G.~A. Hedlund},
\newblock \bibinfo{title}{Symbolic dynamics {II. Sturmian} trajectories},
\newblock \bibinfo{journal}{Amer. J. Math.} \bibinfo{volume}{62}
  (\bibinfo{year}{1940}) \bibinfo{pages}{1--42}.
\bibitem[{Inc.(2011)}]{OEIS}
\bibinfo{author}{OEIS Foundation Inc.}, \bibinfo{title}{The On-Line Encyclopedia of
  Integer Sequences}, \bibinfo{note}{\url{http://oeis.org/}}, \bibinfo{year}{2011}.

\end{thebibliography}

\end{document}